\documentclass[11pt, titlepage]{article}
\usepackage[sort&compress, round, semicolon, authoryear]{natbib}
\usepackage{graphicx, lscape}
\usepackage{graphics, color}
\usepackage{epsfig}
\usepackage{epstopdf}
\usepackage{wrapfig}
\usepackage{caption}
\usepackage{setspace}
\usepackage{subcaption}
\usepackage{amsmath, amsthm, amssymb, amscd}
\usepackage{latexsym}
\usepackage{multirow} 
\usepackage{array}
\usepackage{color}
\usepackage{rotating}
\usepackage[hyphens]{url}
\usepackage[hidelinks]{hyperref}  
\usepackage{cleveref}
\usepackage{fullpage}
\usepackage{fancyvrb}
\usepackage{pdfpages}
\usepackage{fmtcount}
\usepackage[margin=3.3cm]{geometry}
\usepackage{verbatim}
\usepackage[english]{babel}
\usepackage[utf8]{inputenc}
\usepackage{tikz} 
\usetikzlibrary{arrows, decorations.pathmorphing, backgrounds, fit, positioning, shapes.symbols, chains}
\usetikzlibrary{decorations.markings}
\usepackage{lscape}
\usepackage{algpseudocode}
\usepackage{algorithm}
\usepackage{soul}
\usepackage{booktabs}
\usepackage{enumerate}
\usepackage{colortbl}
\long\def\comment#1{} 
\usepackage{booktabs}
\usepackage{yhmath}  % For notation of geodesic segment joining A and B, \wideparen{AB}
\usepackage{mathtools}
\usepackage{rotating}
\usepackage{bm,upgreek}
\usepackage{mathrsfs}
\usepackage{changepage}
\usepackage{enumitem}
\usepackage{comment}
\usepackage[title,titletoc]{appendix}
\usepackage{kotex}

\usepackage{xcolor}
\definecolor{jm}{rgb}{0.5, 0.1, 0.3}
\definecolor{jh}{rgb}{0.2, 0.1, 0.6}
\definecolor{dark_green}{rgb}{0.1, 0.5, 0.1}

\newcommand{\romnum}[1]{\lowercase\expandafter{\romannumeral #1\relax}}

\newtheorem{proposition}{Proposition}[section]
\newtheorem{corollary}{Corollary}[section]
\newtheorem{theorem}{Theorem}[section]
\newtheorem{lemma}{Lemma}[section]

\theoremstyle{definition}
\newtheorem{definition}{Definition}[section]
\newtheorem{remark}{Remark}[section]
\DeclareMathOperator*{\argmin}{arg\,min}

\usepackage{mathtools}

\newcommand{\bdm}{\begin{displaymath}}
	\newcommand{\edm}{\end{displaymath}}

\newcommand{\bs}{\boldsymbol}

\newcommand{\llangle}{\langle\!\langle}
\newcommand{\rrangle}{\rangle\!\rangle}

\begin{document}
\newpage
\begin{center}
{\Large Cross-Spectral Analysis of Bivariate Graph Signals
\medskip
}
\vskip 7mm

{\sc Kyusoon Kim and Hee-Seok Oh}\\
{Seoul National University, Seoul 08826, Korea}
%\end{center}
\vskip 5mm

\today
\end{center}
\vskip 5mm
\begin{quote}
\begin{center}
\textbf{Abstract}
\end{center}

With the advancements in technology and monitoring tools, we often encounter multivariate graph signals, which can be seen as the realizations of multivariate graph processes, and revealing the relationship between their constituent quantities is one of the important problems. To address this issue, we propose a cross-spectral analysis tool for bivariate graph signals. The main goal of this study is to extend the scope of spectral analysis of graph signals to \textcolor{black}{bivariate} graph signals. In this study, we define joint weak stationarity graph processes and introduce graph cross-spectral density and coherence for \textcolor{black}{bivariate} graph processes. We propose several estimators for the cross-spectral density and investigate the theoretical properties of the proposed estimators. Furthermore, we demonstrate the effectiveness of the proposed estimators through numerical experiments, including simulation studies and a real data application. Finally, as an interesting extension, we discuss robust spectral analysis of graph signals in the presence of outliers.

\end{quote}	
	
	\pagenumbering{arabic}
	
\section{Introduction}
Graph signals, defined on nodes (or vertices) in a graph (or network), have become popular in various fields, such as signal processing, environmental science, economics, and sociology. 

Many studies have focused on developing deterministic graph signal processing tools, covering topics such as sampling and recovery \cite{Narang2011, Chen2015a, Chen2015b, Tanaka2018, Lorenzo2018}, graph filter and filter bank design \cite{Narang2012, Tremblay2018, Sakiyama2014, Anis2017}, representation \cite{Crovella2003, Narang2009, Hammond2011, Sandryhaila2013GFT,  Bahonar2021graph, Bulai2023}, denoising \cite{Onuki2016, Waheed2018, Stankovic2019, Zhou2024graph}, \textcolor{black}{data segmentation \cite{Giraldo2021emerging, Giraldo2022graph, Deng2024body}, and application to point clouds \cite{Wei2022geometry, Yang2022inferring, Liu2024point}}. More recently, statistical graph signal processing, which uses a probabilistic approach to deal with random graph processes, has gained attention. In statistical graph signal processing, as in classical signal processing, the notion of stationarity is fundamental and has been defined on graphs in several works of literature \cite{Girault2015, Segarra2016, Perraudin2017}. Stationary graph processes can be characterized by their graph power spectral density in the graph frequency domain, which can facilitate understanding, model building, and prediction of these processes \cite{Yang2022recommendation}.

Spectral analysis using graph power spectral density is typically limited to univariate graph processes. However, we often encounter multivariate graph processes in the real world, such as RGB signals on pixels, meteorological data from environmental sensor networks, economic indicators in trading networks, and brain network responses to stimuli. \textcolor{black}{These processes have gotten more and more attention recently \cite{Kurokawa2017, Yan2022, Wei2024generalized}.} For these multivariate graph signals, we are interested in investigating interrelationships between the processes, such as their correlation and shared information, which are not readily apparent from analyzing individual spectra. In classical signal processing, this type of analysis is called cross-spectral analysis \cite{Priestley1982, Von2002, Fuller2009}. The main contribution of this study is the development of a cross-spectral analysis framework on graphs. Specifically, we extend the works of \cite{Segarra2016} and \cite{Perraudin2017} to analyze \textcolor{black}{bivariate} graph processes. We define the notions of joint weak stationarity and cross-spectral density on graphs. We then propose several nonparametric estimators for the graph cross-spectral density, \textcolor{black}{along with a robust version,} and provide theoretical properties. \textcolor{black}{Compared to the works of \cite{Segarra2016} and \cite{Perraudin2017}, the proposed method not only identifies the shared information between two graph signals but also lays the groundwork for taking the first step to advance the analysis of multivariate graph signals. We further demonstrate} the effectiveness of the proposed estimators through numerical experiments.

The remainder of this paper is organized as follows. In Section \ref{sec:prelim}, we review the graph spectral theory and the concepts of weakly stationary graph processes and power spectral density on graphs. Several nonparametric estimators for the graph power spectral density are also introduced in this section. In Section \ref{sec:cross_spectrum_analysis_graph}, we define the notions of joint weak stationarity, cross-spectral density, and coherence on graphs. In addition, we propose estimators for the cross-spectral density and investigate their theoretical properties. Section \ref{sec:numerical}  conducts numerical experiments on various graphs. In Section \ref{sec:robustspectral}, we discuss a robust spectral analysis on graphs, proposing $M$-type estimators for the graph power spectral density and graph cross-spectral density.  Concluding remarks are provided in Section \ref{sec:conclusion}. The proofs of theoretical properties in Section \ref{sec:cross_spectrum_analysis_graph} are given in Appendix. {\tt R} codes to reproduce the numerical experiments can be available at \url{https://github.com/qsoon/graph-cross-spectral-analysis}.

\section{Preliminary} \label{sec:prelim}
\subsection{Graph spectral theory and notations} \label{sec:gspnotation}
Let $\mathcal{G}=(\mathcal{V}, \mathcal{E}, \bs{A})$ be an undirected or a directed weighted graph with a set of $N$ nodes $\mathcal{V}$ (i.e., $\lvert \mathcal{V} \rvert = N$), a set of edges $\mathcal{E}$, and a weighted adjacency matrix $\bs{A}=(A_{i,j})_{N\times N}$ whose entry $A_{i,j} \neq 0$ only if the $i$th node and the $j$th node are connected.  Denote a graph shift operator (GSO) by a matrix $\bs{S}=(S_{i,j})_{N\times N}$, which is a local operator that replaces a signal value $x_i$ at the $i$th node of a graph signal $\bs{x}=(x_1, \ldots, x_N)^\top \in \mathbb{R}^N$ with a linear combination of signal values at neighboring nodes of the $i$th node \cite{Segarra2018, Sandryhaila2013, Sandryhaila2014}. The entry of GSO, $S_{i,j}$, is nonzero only if $i=j$ or the $i$th node and the $j$th node are connected, so GSO contains information on the local structure of graph $\mathcal{G}$. The weighted adjacency matrix $\bs{A}$ and the graph Laplacian $\bs{L}:=\bs{D}-\bs{A}$ (where $\bs{D}$ is a diagonal matrix with $D_{i,i} = \sum_{j=1}^N A_{i,j}$) are examples of GSO. If $\bs{S}$ is normal (i.e., $\bs{S}^H \bs{S}= \bs{S} \bs{S}^H$, where $\bs{S}^H$ denotes the complex conjugate transpose of $\bs{S}$; \textcolor{black}{see \cite{Horn1985}}), the spectral theorem states that there exists a unitary matrix $\bs{V}$ (i.e., $\bs{V}^H \bs{V} = \bs{V} \bs{V}^H = \bs{I}$) such that $\bs{S} = \bs{V} \bs{\Lambda} \bs{V}^H$ with a diagonal matrix $\bs{\Lambda}$. The entries of  $\bs{\Lambda}$ need not be real, whereas they are real if $\bs{S}$ is Hermitian (i.e., $\bs{S}^H = \bs{S}$). In this study, we only consider real-valued graph signals and normal GSOs. 

Given a normal GSO $\bs{S}$ and a graph signal $\bs{x} \in \mathbb{R}^N$, the graph Fourier transform (GFT) is defined as $\tilde{\bs{x}} = \bs{V}^H \bs{x}$, which is referred to as the graph Fourier coefficients or spectrum coefficients of $\bs{x}$ \cite{Sandryhaila2013GFT}. GFT represents a graph signal in terms of frequencies by transforming the signal from the node domain to the graph spectral domain. Throughout this paper, $\langle \cdot \rangle$, $\circ$, $\llangle \cdot, \cdot \rrangle$, and $*$ denote the inner product, the elementwise product, the elementwise inner product, and the elementwise complex conjugate, respectively. That is, $\llangle \bs{\xi}, \bs{\zeta} \rrangle = \bs{\xi} \circ \bs{\zeta}^*$ for $\bs{\xi}$ and $\bs{\zeta}$, which are complex-valued vectors. Furthermore, $\lvert \cdot \rvert^2$ is applied elementwise.

\subsection{Weakly stationary graph processes} \label{sec:stationarity}
In statistical graph signal processing, we consider a graph signal $\bs{x}$ to be a realization of a random graph process $X$. Since weak stationarity plays a vital role in statistical signal processing, it is a priority to establish the notion of weak stationarity for graph processes. Defining weak stationarity on graphs is challenging due to the absence of directional concepts, making it difficult to define notions such as translation and lag. Several studies have been done to define the weak stationarity of graph processes. Girault \cite{Girault2015} and \cite{Girault2015thesis} pioneered a definition of weak stationarity for graph processes based on an isometric graph translation operator, which preserves the signal's energy and is isometric, unlike the GSOs. They defined a weakly stationary graph process as one whose first and second moments remain invariant under multiplication by an isometric translation. However, this new translation operator does not have localization properties essential to the stationarity concept.

Perraudin and Vandergheynst \cite{Perraudin2017} proposed an alternative definition of weakly stationary graph processes using graph Laplacian as a localization operator on graphs. Segarra et al. \cite{Segarra2018} and Marques et al. \cite{Marques2017} extended this definition, generalizing it to normal GSOs, thus not restricting it to the graph Laplacian. They suggested three equivalent definitions (under mild conditions) as follows.

\begin{definition} \label{def:weakstationary}
	\cite{Marques2017} Given a normal graph shift operator $\bs{S}$, a zero-mean random graph process $X$ is weakly stationary with respect to $\bs{S}$ if one of the following holds:
	\begin{itemize}
		\item[(a)] The graph process $X$ can be expressed as the response of the linear shift-invariant graph filter $\bs{\mathrm{H}} = \sum_{\ell=0}^{N-1} h_\ell \bs{S}^\ell$ to a zero-mean white input $\bs{\epsilon}$ (i.e., $E(\bs{\epsilon}) = \bs{0}$ and $E\left(\bs{\epsilon}\bs{\epsilon}^H\right) = \bs{I}$). 
		\item[(b)] For any set of non-negative integers $a$, $b$, and $c\le b$, it holds that
		\begin{equation*}
			E\Big((\bs{S}^a X)((\bs{S}^H)^b X)^H\Big) = E\Big((\bs{S}^{a+c} X)((\bs{S}^H)^{b-c} X)^H\Big)
		\end{equation*}
		\item[(c)] The covariance matrix $\bs{\Sigma}_{X} = E\left(XX^H\right)$ and $\bs{S}$ are simultaneously diagonalizable.
	\end{itemize}
\end{definition}

\begin{remark}
	It is known that a graph filter is linear shift-invariant if and only if it is a polynomial in the graph shift \cite{Sandryhaila2013}. It is also worth noting that \Cref{def:weakstationary}(a) gives a constructive definition of weak stationarity. \Cref{def:weakstationary}(b) and (c) generalize the definition of weak stationarity in classical signal processing. The former corresponds to the requirement that the second moment of a weakly stationary process is invariant to time shifts, and the latter corresponds to the requirement that the covariance of a weakly stationary process must be circulant.	
\end{remark}

\begin{remark} \label{rmk:notzeromeanWS}
	If $X$ is not a zero-mean process, it is required to have the first moment of the form $E(X) = \alpha \bs{v}_k$ for a constant $\alpha$ and $1\le k \le N$ to satisfy weak stationarity with respect to $\bs{S}= \bs{V} \bs{\Lambda} \bs{V}^H$ for $\bs{V} = \left(\bs{v}_1|\cdots| \bs{v}_N\right)$ \cite{Marques2017,Segarra2016}. If $\bs{S} = \bs{L}$ and $k=1$, $E(X)$ becomes a constant signal. In other literature, Girault \cite{Girault2015} required $E(X)$ to be a zero vector, while \cite{Perraudin2017} required it to be  $\alpha \bs{1}$.
\end{remark}

\begin{proposition} \label{prop:equivdef}
	\cite{Marques2017}  \Cref{def:weakstationary}(b) and (c) are equivalent. Furthermore, if the eigenvalues of $\bs{S}$ are all distinct, all three definitions in \Cref{def:weakstationary} are equivalent.
\end{proposition}
\begin{proof}
	The proofs can be found in \cite{Marques2017}.
\end{proof}

\subsection{Graph power spectral density} \label{sec:gpsd}
To describe the energy distribution of a weakly stationary graph process over frequency, we can define the {\it graph power spectral density} (GPSD) in \Cref{def:gpsd} \cite{Segarra2016,Perraudin2017,Marques2017}. We assume that GSO $\bs{S}$ is normal, and the appearing graph processes are real-valued and zero-mean. For consistency, we will continue to use the notation $H$ or $*$ when dealing with real-valued vectors.
\begin{definition} \label{def:gpsd}
	\cite{Marques2017} The graph power spectral density of a zero-mean weakly stationary graph process $X$ is defined as a non-negative $N \times 1$ vector $\bs{p}_{X} = \operatorname{diag}(\bs{V}^H \bs{\Sigma}_{X} \bs{V})$.
\end{definition}

\begin{remark}
	If $X$ is weakly stationary with respect to $\bs{S} = \bs{V} \bs{\Lambda} \bs{V}^H$, the covariance matrix $\bs{\Sigma}_{X}$ and $\bs{S}$ are simultaneously diagonalizable by \Cref{def:weakstationary}(c). Therefore, we can rewrite $\bs{\Sigma}_{X} = \bs{V} \operatorname{diag}(\bs{p}_{X}) \bs{V}^H$. Moreover, since $\bs{\Sigma}_{X}$ is Hermitian, $\bs{p}_{X}$ is a real-valued vector.
\end{remark}

We note that $\bs{p}_{X} = \operatorname{diag}(\bs{V}^H \bs{\Sigma}_{X} \bs{V}) = E(\lvert \bs{V}^H X\rvert^2)= E(\lvert \tilde{X}\rvert^2)$, where $\tilde{X}$ denotes the GFT of $X$. This indicates that GPSD characterizes the energy distribution over graph frequencies. Thus, we can interpret the $\ell$th element of $\bs{p}_{X}$ as $(\bs{p}_{X})_\ell =: \bs{p}_{X}(\lambda_\ell)$, where $\lambda_\ell$ ($1 \le \ell \le N)$ are the eigenvalues of $\bs{S}$. %Here, the notation $\bs{p_{\mathrm{x}}}$ is used interchangeably: on the left-hand side, it represents an $N \times 1$ vector, while on the right-hand side, it denotes a function defined in the graph spectral domain.

\subsection{Nonparametric estimation of graph power spectral density} \label{sec:gpsd_estimate}
%	In this section, we introduce several nonparametric estimators for the GPSD, proposed by \cite{Marques2017} and \cite{Segarra2018}, which do not assume any specific parametric model to generate a graph process $\bs{\mathrm{x}}$.

\subsubsection{Graph periodogram} \label{sec:graph_periodogram}
Assuming that we have $R$ realizations $\{\bs{x}_r\}_{r=1}^{R}$ of a weakly stationary graph process $X$ in $\bs{S}=\bs{V} \bs{\Lambda} \bs{V}^H$, Marques et al. \cite{Marques2017} proposed three nonparametric estimators of GPSD that follow naturally from \Cref{def:gpsd}:  periodogram, correlogram, and least squares types. The periodogram-type estimator is defined as
\begin{equation*} \label{eq:graph_periodogram}
	\hat{\bs{p}}_{X}^{p} = \frac{1}{R} \sum\limits_{r=1}^{R} \lvert \tilde{\bs{x}}_r \rvert^2 = \frac{1}{R} \sum\limits_{r=1}^{R} \lvert \bs{V}^H \bs{x}_r \rvert^2,
\end{equation*}
where  $\tilde{\bs{x}}_r$ denotes the graph Fourier coefficients of $\bs{x}_r$. The correlogram-type estimator is defined as
\begin{equation*} \label{eq:graph_correlogram}
	\hat{\bs{p}}_{X}^{c} = \operatorname{diag}(\bs{V}^H \hat{\bs{\Sigma}}_{X} \bs{V}),
\end{equation*}
where $\hat{\bs{\Sigma}}_{X} = \frac{1}{R} \sum_{r=1}^{R} \bs{x}_r \bs{x}_r^H$. These two estimators are analogous to the periodogram and correlogram of time signals in classical signal processing. Finally, for $\bs{G} = \left(\operatorname{vec}(\bs{v}_1 \bs{v}_1^H) | \cdots | \operatorname{vec}(\bs{v}_N \bs{v}_N^H)\right)$, the least squares estimator is defined as
\begin{equation} \label{eq:graph_ls}
	\hat{\bs{p}}_{X}^{ls} = \argmin_{\bs{p}} \lVert \hat{\bs{\sigma}}_{X} - \bs{G} \bs{p} \rVert_2^2 = (\bs{G}^H \bs{G})^{-1} \bs{G}^H \hat{\bs{\sigma}}_{X},
\end{equation}
where $\hat{\bs{\sigma}}_{X} = \operatorname{vec}(\hat{\bs{\Sigma}}_{X})$.

\begin{proposition} \label{prop:graph_periodogram_equiv}
	\cite{Segarra2018} For a zero-mean weakly stationary graph process $X$, the three GPSD estimators are identical, i.e.,
	\begin{equation*}
		\hat{\bs{p}}_{X}^{p} = \hat{\bs{p}}_{X}^{c} = \hat{\bs{p}}_{X}^{ls}.
	\end{equation*}
\end{proposition}
%	\begin{proof}
	%		This is a special case of \Cref{prop:graph_crossperiodogram_equiv} when $\bs{\mathrm{x}} = \bs{\mathrm{y}}$. Therefore, the proof follows from the proof of \Cref{prop:graph_crossperiodogram_equiv}.
	%	\end{proof}

% \begin{proof}
	%     We show that the $i$th element of $\hat{\bs{p}}_{pg}$ is equal to that of $\hat{\bs{p}}_{cg}$.
	%     \begin{equation*}
		%         (\hat{\bs{p}}_{cg})_i = \left(\bs{V}^H\left(\frac{1}{R}\sum\limits_{r=1}^{R} \bs{x}_r \bs{x}_r^H\right) \bs{V}\right)_{i,i} = \frac{1}{R}\sum\limits_{r=1}^{R} \bs{v}_i^H \bs{x}_r \bs{x}_r^H \bs{v}_i = \frac{1}{R}\sum\limits_{r=1}^{R} \lvert \bs{v}_i^H \bs{x}_r \rvert^2 = (\hat{\bs{p}}_{pg})_i
		%     \end{equation*}
	%     Therefore, $\hat{\bs{p}}_{pg} = \hat{\bs{p}}_{cg}$. 
	
	%     Now, we show that the $i$th element of $\hat{\bs{p}}_{cg}$ is equal to that of $\hat{\bs{p}}_{ls}$.
	%     It can be easily shown that $\bs{G}^H \bs{G} = I$. Thus, $\hat{\bs{p}}_{ls} = \bs{G}^H \hat{\bs{\sigma}}_{\bs{x}}$. Then, 
	%     \begin{equation*}
		%         (\hat{\bs{p}}_{ls})_i = (\bs{G}^H \hat{\bs{\sigma}}_{\bs{x}})_i = [v_{i1}\bs{v}_i^H, \ldots,v_{iN}\bs{v}_i^H]\hat{\bs{\sigma}}_{\bs{x}} = \bs{v}_i^H \hat{\bs{\Sigma}}_{\bs{x}} \bs{v}_i = (\hat{\bs{p}}_{cg})_i.
		%     \end{equation*}
	%     Thus, $\hat{\bs{p}}_{cg} = \hat{\bs{p}}_{ls}$. The proof is completed. 
	% \end{proof}

Hence, we refer to the three estimators as \textit{graph periodogram} and denote them by $\hat{\bs{p}}_{X}$.

\subsubsection{Windowed average graph periodogram} \label{sec:windowagvgperiodogram}
It is known that the bias of the graph periodogram $\hat{\bs{p}}_{X}$ is zero, but its covariance is equal to $(2/R)\operatorname{diag}^2(\bs{p}_X)$ when the process $X$ is Gaussian and $\bs{S}$ is symmetric, resulting in MSE of the graph periodogram being $(2/R)\lVert \bs{p}_X \rVert_2^2$ \cite{Marques2017}. Therefore, the number of realizations $R$ significantly affects the estimation performance. As $R$ increases, the graph periodogram tends to get closer to the true GPSD, while small $R$ degrades the estimation performance. However, obtaining multiple realizations of the graph process is often impractical. In classical signal processing, the issue of having only one or a few realizations can be effectively tackled using data windows, such as Bartlett's method \cite{Bartlett1950} and Welch's method \cite{Welch1967}, or lag windows, such as the Blackman--Tukey method \cite{Blackman1959}. Similarly, for graph signals, windows can be utilized to estimate GPSD when only one or a few realizations are available.

Marques et al. \cite{Marques2017} proposed a windowed average graph periodogram when one realization $\bs{x}$ is available. They introduced a data window $\bs{w} \in \mathbb{R}^N$ satisfying $\lVert \bs{w} \rVert_2^2 = \lVert \bs{1} \rVert_2^2 = N$, and then defined a windowed graph periodogram, for a windowed signal $\bs{x_w} = \operatorname{diag}(\bs{w}) \bs{x}$, as
\begin{equation*} \label{eq:window_graph_periodogram}
	\hat{\bs{p}}_{X}^{\bs{w}} = \lvert \bs{V}^H \bs{x_w} \rvert^2 = \lvert \tilde{\bs{W}} \bs{V}^H \bs{x} \rvert^2, 
\end{equation*}
where $\tilde{\bs{W}} = \bs{V}^H \operatorname{diag}(\bs{w}) \bs{V}$. To generate multiple signals from one realization, they utilized a bank of $M$ windows, $\mathcal{W} = \{\bs{w}_m\}_{m=1}^{M}$. A windowed average graph periodogram is then defined as
\begin{equation} \label{eq:window_avg_graph_periodogram}
	\hat{\bs{p}}_{X}^{\mathcal{W}} = \frac{1}{M} \sum\limits_{m=1}^{M} \hat{\bs{p}}_{X}^{\bs{w}_m}.
\end{equation}

\section{Cross-spectral analysis on graphs} \label{sec:cross_spectrum_analysis_graph}
%	If we have a multivariate graph signal, a cross-spectral analysis would be very helpful. 
\subsection{Jointly weakly stationary graph processes} \label{sec:joint_weak_stationarity_graph}
First, we define the notion of joint weak stationarity on graphs for cross-spectral analysis. \textcolor{black}{Similar to how joint weak stationarity is defined in classical signal processing, this concept can be extended from  \Cref{def:weakstationary} as follows.} 

\begin{definition} \label{def:jointweakstationary}
	Given a normal graph shift operator $\bs{S}=\bs{V} \bs{\Lambda} \bs{V}^H$, if zero-mean random graph processes $X$ and $Y$ are both weakly stationary with respect to $\bs{S}$ and one of the following holds, then $X$ and $Y$  are jointly weakly stationary with respect to $\bs{S}$: 
	\begin{itemize}
		\item[(a)] Graph processes $X$ and $Y$ can be expressed as the responses of two linear shift-invariant graph filters $\bs{\mathrm{H}}_1 = \sum_{\ell=0}^{N-1} h_{1,\ell} \bs{S}^\ell$ and $\bs{\mathrm{H}}_2=\sum_{\ell=0}^{N-1} h_{2,\ell} \bs{S}^\ell$ to zero-mean white inputs $\bs{\epsilon}_1$ and $\bs{\epsilon}_2$, where the cross-covariance satisfies  $\bs{V}^H Cov(\bs{\epsilon}_1, \bs{\epsilon}_2) \bs{V}$ being diagonal. 
		\item[(b)] For any set of non-negative integers $a$, $b$, and $c \le b$, it holds that
		\begin{equation*}
			E\Big((\bs{S}^a X) ((\bs{S}^H)^b Y)^H \Big) = E\Big((\bs{S}^{a+c} X) ((\bs{S}^H)^{b-c} Y)^H \Big)
		\end{equation*}
		\item[(c)] The cross-covariance matrix $\bs{\Sigma}_{XY} = E(XY^H)$ and $\bs{S}$ are simultaneously diagonalizable. 
	\end{itemize}
\end{definition}

\begin{remark} \label{rmk:loukaswork}
	We note that \Cref{def:jointweakstationary} differs from the notion of jointly wide-sense stationarity proposed by \cite{Loukas2019}, which dealt with time-varying processes defined over graphs. Their terms ``jointly wide-sense stationarity" and ``time-vertex wide-sense stationarity" imply that a process is multivariate time-wide-sense stationary and multivariate vertex-wide-sense stationary, respectively. Here, the definition of multivariate vertex-wide-sense stationarity is the same as \Cref{def:jointweakstationary}(c); however, \Cref{def:jointweakstationary} provides additional equivalent definitions, including \Cref{def:jointweakstationary}(a), which offers a constructive definition of joint weak stationarity. Moreover, in their work, the concept of joint power spectral density as a function of graph frequency and angular frequency represents the energy distribution in the joint frequency domain. This differs from the graph cross-spectral density described in \Cref{sec:gcpsd}, which represents the energy distribution between two signals as a function of graph frequency.	
\end{remark}

\begin{remark} \label{rmk:twocrosscov}
	Since $\bs{\Sigma}_{XY}=(\bs{\Sigma}_{YX})^H$, \Cref{def:jointweakstationary}(c) is equivalent to that $\bs{\Sigma}_{YX}$ and $\bs{S}$ are simultaneously diagonalizable.
\end{remark}

\begin{remark} \label{rmk:notzeromeanJWS}
	Similar to \Cref{rmk:notzeromeanWS}, if $X$ and $Y$ are not zero-mean processes, they are required to have the first moments of the forms $E(X) = \alpha \bs{v}_k$ and $E(Y) = \beta \bs{v}_\ell$ for some constants $\alpha$ and $\beta$ and $1 \le k, \ell \le N$ to satisfy joint weak stationarity with respect to $\bs{S} = \bs{V} \bs{\Lambda} \bs{V}^H$. Recall that the graph processes considered in this study are assumed to be real-valued and zero-mean.  
\end{remark}

\begin{remark} \label{rmk:stationarylevel}
	Perraudin and Vandergheynst \cite{Perraudin2017} and Marques et al. \cite{Marques2017} discussed a measure of weak stationarity as $$\lVert \operatorname{diag}(\bs{V}^H \hat{\bs{\Sigma}}_{X} \bs{V})\rVert_2 / \lVert \bs{V}^H \hat{\bs{\Sigma}}_{X} \bs{V} \rVert_F$$ for the estimated covariance matrix $\hat{\bs{\Sigma}}_{X}$. In a similar vein, the joint weak stationarity level can be measured using the computed measure $\lVert \operatorname{diag}(\bs{V}^H \hat{\bs{\Sigma}}_{XY} \bs{V})\rVert_2 / \lVert \bs{V}^H \hat{\bs{\Sigma}}_{XY} \bs{V} \rVert_F$ for the estimated cross-covariance matrix  $\hat{\bs{\Sigma}}_{XY}$, where $\lVert \cdot \rVert_F$ denotes the Frobenius norm. The closer this measure is to 1, the more the two graph processes are jointly weakly stationary.
\end{remark}

\begin{proposition} \label{prop:jointstationarity_defequiv}
	\Cref{def:jointweakstationary}(b) and (c) are equivalent. Furthermore, if the eigenvalues of $\bs{S}$ are all distinct, all three definitions in \Cref{def:jointweakstationary} are equivalent.
\end{proposition}
A proof is provided in Appendix.

\subsection{Graph cross-spectral density} \label{sec:gcpsd}
Like power spectral density, cross-spectral density (CSD) is a function of frequency. It describes the frequencies shared between two signals and provides information about the correlation of the two signals in the frequency domain. Extending the notion of CSD in classical signal processing, we define a graph cross-spectral density (GCSD) based on the definition of the graph power spectral density.	
\begin{definition} \label{def:gcpsd}
	The graph cross-spectral density of zero-mean weakly stationary graph processes $X$ and $Y$ is defined as an $N \times 1$ vector $\bs{p}_{XY} = \operatorname{diag}(\bs{V}^H \bs{\Sigma}_{XY} \bs{V})$. 
\end{definition}

\begin{remark}
	If $X$ and $Y$ are jointly weakly stationary with respect to $\bs{S}=\bs{V} \bs{\Lambda} \bs{V}^H$, the cross-covariance matrix $\bs{\Sigma}_{XY}$ and $\bs{S}$ are simultaneously diagonalizable by \Cref{def:jointweakstationary}(c). Thus, we can rewrite $\bs{\Sigma}_{XY} = \bs{V} \operatorname{diag}(\bs{p}_{XY}) \bs{V}^H$. Unlike GPSD, $\bs{p}_{XY}$ can be a complex-valued vector because $\bs{\Sigma}_{XY}$ is not required to be Hermitian. 
	
	\noindent In addition, according to \Cref{rmk:twocrosscov}, an alternative definition of GCSD can be formulated as  $\bs{p}_{YX} = \operatorname{diag}(\bs{V}^H \bs{\Sigma}_{YX} \bs{V})$. Since $\bs{\Sigma}_{XY} = (\bs{\Sigma}_{YX})^H$, the relationship between the two definitions is  $\bs{p}_{XY} = \bs{p}^*_{YX}$. As a result, we adopt $\bs{p}_{XY}$ as the definition of GCSD because they contain the same information. 
\end{remark}

\textcolor{black}{We further note that $\bs{p}_{XY} = \operatorname{diag}(\bs{V}^H \bs{\Sigma}_{XY} \bs{V}) = E(\llangle \bs{V}^H X, \bs{V}^H Y\rrangle) = E(\llangle \tilde{X}, \tilde{Y}\rrangle)$, where $\tilde{Y}$ denotes the GFT of $Y$. Therefore, GCSD characterizes the amount of shared information between the two graph processes across graph frequencies. Similar to the GPSD, the $\ell$th element of $\bs{p}_{XY}$ can be interpreted as $(\bs{p}_{XY})_\ell =: \bs{p}_{XY}(\lambda_\ell)$, representing the cross-spectral density at graph frequency $\lambda_\ell$.}

We present three propositions analogous to Properties 1, 2, and 3 in \cite{Segarra2016}.

\begin{proposition} \label{prop:spectralconvolution_cpsd}
	Let $X$ and $Y$ be zero-mean jointly weakly stationary processes with cross-covariance $\bs{\Sigma}_{XY}$ and GCSD $\bs{p}_{XY}$. Consider two filters $\bs{\mathrm{H}}_1$ and $\bs{\mathrm{H}}_2$ with frequency responses $\tilde{\bs{h}}_1 = \operatorname{diag}(\bs{V}^H \bs{\mathrm{H}}_1 \bs{V})$ and $\tilde{\bs{h}}_2 = \operatorname{diag}(\bs{V}^H \bs{\mathrm{H}}_2 \bs{V})$, respectively, and define $Z = \bs{\mathrm{H}}_1 X$ and $W = \bs{\mathrm{H}}_2 Y$. Then, 
	\begin{itemize}
		\item[(a)] $Z$ and $W$ are jointly stationary in $\bs{S}$ with cross-covariance $\bs{\Sigma}_{ZW} = \bs{\mathrm{H}}_1 \bs{\Sigma}_{XY} \bs{\mathrm{H}}_2^H$. 
		\item[(b)] $Z$ and $W$ have GCSD $\bs{p}_{ZW} =  \tilde{\bs{h}}_1  \circ \tilde{\bs{h}}_2^*  \circ \bs{p}_{XY}$.
	\end{itemize}
\end{proposition}
% \begin{proof}
	%     We have $\bs{\mathrm{H}}_1 = \bs{V} \operatorname{diag}(\tilde{\bs{h}}_1) \bs{V}^H$, $\bs{\mathrm{H}}_2 = \bs{V} \operatorname{diag}(\tilde{\bs{h}}_2) \bs{V}^H$, and $\bs{\Sigma}_{\bs{\mathrm{xy}}} = \bs{V}\operatorname{diag}(\bs{p}_{\bs{\mathrm{xy}}}) \bs{V}^H$. By using these three equations, we can derive 
	%     \begin{equation*}
		%         \bs{\Sigma}_{\bs{\mathrm{zw}}} = \bs{\mathrm{H}}_1 \bs{\Sigma}_{\bs{\mathrm{xy}}} \bs{\mathrm{H}}_2^H = \bs{V} \operatorname{diag}(\tilde{\bs{h}}_1) \operatorname{diag}(\bs{p}_{\bs{\mathrm{xy}}}) \operatorname{diag}(\tilde{\bs{h}}_2^*) \bs{V}^H.
		%     \end{equation*} It completes the proofs for (a) and (b).
	% \end{proof}

\begin{proposition} \label{prop:crosscov_element}
	Let $X$ and $Y$ be zero-mean random graph processes that are outputs of linear shift-invariant graph filters $\bs{\mathrm{H}}_1$ and $\bs{\mathrm{H}}_2$ of degree $L_1-1$ and $L_2-1$, respectively, applied to white inputs $\bs{\epsilon}_1$ and $\bs{\epsilon}_2$, where $Cov(\bs{\epsilon}_1, \bs{\epsilon}_2)=\bs{I}$. Then, if the distance between the $i$th node and the $j$th node is greater than $(L_1+L_2-2)$ hops, $(\bs{\Sigma}_{XY})_{i,j}=0$.
\end{proposition}
% \begin{proof}
	%     Since $\bs{\mathrm{H}}_1$ and $\bs{\mathrm{H}}_2$ are linear shift-invariant, by \cite{Sandryhaila2013}, we can rewrite $\bs{\mathrm{H}}_1$ and $\bs{\mathrm{H}}_2$ as
	%     \begin{equation*}
		%         \bs{\mathrm{H}}_i = \sum\limits_{\ell=0}^{L_i-1} h_{i,\ell} \bs{S}^\ell, \quad i=1,2.
		%     \end{equation*}
	%     Since $\bs{S}_{i,j}$ is nonzero only if $i=j$ or the $i$th node and the $j$th node are connected, it is clear that $\bs{S}^\ell_{i,j}$ is nonzero only if the distance between the $i$th node and the $j$th node is equal to or less than $\ell$. Based on the fact that $\bs{\Sigma_{\mathrm{xy}}} = \bs{\mathrm{H}}_1 \bs{\mathrm{H}}_2^H$, and that it is a polynomial of $\bs{S}$ and $\bs{S}^H$ of order $L_1 + L_2 -2$, the proof is completed.
	% \end{proof}

\begin{proposition} \label{prop:gftuncorr_cpsd}
	Given zero-mean jointly weakly stationary graph processes $X$ and $Y$ with GCSD $\bs{p}_{XY}$, define GFT processes as $\tilde{X} = \bs{V}^H X$ and $\tilde{Y} = \bs{V}^H Y$. Then, $\tilde{X}$ and $\tilde{Y}$ are uncorrelated, and its cross-covariance is 
	\[
	\bs{\Sigma}_{\tilde{X}\tilde{Y}} = E(\tilde{X} \tilde{Y}^H) = \operatorname{diag}(\bs{p}_{XY}).
	\]
\end{proposition}
% \begin{proof}
	%     We can compute the covariance matrix of $\tilde{\bs{\mathrm{x}}}$ as following and the proof is completed.
	%     \begin{equation*}
		%         \bs{\Sigma}_{\tilde{\bs{\mathrm{x}}}\tilde{\bs{\mathrm{y}}}} = E[\tilde{\bs{\mathrm{x}}} \tilde{\bs{\mathrm{y}}}^H] = \bs{V}^H E[\bs{\mathrm{x}}\bs{\mathrm{y}}^H]\bs{V} = \bs{V}^H \bs{\Sigma}_{\bs{\mathrm{x}}\bs{\mathrm{y}}}\bs{V} = \operatorname{diag}(\bs{p_{\mathrm{xy}}}).
		%     \end{equation*}
	% \end{proof}

\Cref{prop:spectralconvolution_cpsd} is the counterpart of the spectral convolution theorem for graph processes. \Cref{prop:crosscov_element} represents a limitation of cross-correlations between signal values generated by specific filtering, which is beneficial in designing windows for the windowed average graph cross-periodogram in \Cref{sec:windowavgcrossperiodogram} when combined with the results of \Cref{prop:gcpsd_windowdesign}. In addition, \Cref{prop:gftuncorr_cpsd} highlights the usefulness of graph frequency domain analysis. Proofs for the above propositions are found in Appendix. 

\subsection{Nonparametric estimation of graph cross-spectral density} \label{sec:gcpsd_estimate}
%We propose several nonparametric estimators of GCSD and further investigate their theoretical properties, generalizing those of \cite{Marques2017} and \cite{Segarra2018}. 	

\subsubsection{Graph cross-periodogram} \label{sec:graph_crossperiodogram}
We propose three nonparametric estimators of GCSD, assuming that we have $R$ realizations $\bs{x}_1, \ldots, \bs{x}_R$ and $\bs{y}_1, \ldots, \bs{y}_R$ of jointly weakly stationary processes $X$ and $Y$, respectively. To emphasize the correspondence between the proposed estimates and those introduced in  \Cref{sec:graph_periodogram}, we present periodogram, correlogram, and least squares types. A periodogram-type estimator is defined as
\begin{equation*} \label{graph_cross_periodogram}
	\hat{\bs{p}}_{XY}^{p} = \frac{1}{R} \sum \limits_{r=1}^{R} \llangle \tilde{\bs{x}}_r ,  \tilde{\bs{y}}_r \rrangle = \frac{1}{R} \sum \limits_{r=1}^{R} (\bs{V}^H \bs{x}_r) \circ (\bs{V}^H \bs{y}_r)^*,
\end{equation*}
where $\tilde{\bs{x}}_r$ and $\tilde{\bs{y}}_r$ are graph Fourier coefficients of $\bs{x}_r$ and $\bs{y}_r$, respectively. A correlogram-type estimator is defined as
\begin{equation*} \label{graph_cross_correlogram}
	\hat{\bs{p}}_{XY}^{c} = \operatorname{diag}(\bs{V}^H \hat{\bs{\Sigma}}_{XY} \bs{V}),
\end{equation*}
where the sample cross-covariance $\hat{\bs{\Sigma}}_{XY} =  \frac{1}{R} \sum_{r=1}^{R} \bs{x}_r \bs{y}_r^H$. Finally, considering $\bs{\Sigma}_{XY} = \bs{V} \operatorname{diag}(\bs{p}_{XY}) \bs{V}^H = \sum_{i=1}^{N} (\bs{p}_{XY})_i \bs{v}_i \bs{v}_i^H$, we define a least squares estimator as
\begin{equation} \label{eq:graph_cross_ls}
	\hat{\bs{p}}_{XY}^{ls} = \argmin_{\bs{p}} \lVert \hat{\bs{\sigma}}_{XY} - \bs{G} \bs{p} \rVert_2^2 = (\bs{G}^H \bs{G})^{-1} \bs{G}^H \hat{\bs{\sigma}}_{XY},
\end{equation}
where $\bs{G} = \left(\operatorname{vec}(\bs{v}_1 \bs{v}_1^H) | \cdots | \operatorname{vec}(\bs{v}_N \bs{v}_N^H)\right)$ and $\hat{\bs{\sigma}}_{XY} = \operatorname{vec}(\hat{\bs{\Sigma}}_{XY})$.

\begin{proposition} \label{prop:graph_crossperiodogram_equiv}
	For zero-mean jointly weakly stationary graph processes $X$ and $Y$, the three different GCSD estimators are identical, i.e.,
	\begin{equation*}
		\hat{\bs{p}}_{XY}^{p} = \hat{\bs{p}}_{XY}^{c} = \hat{\bs{p}}_{XY}^{ls}.
	\end{equation*}
\end{proposition}

% \begin{proof}
	%     Let us denote the $i$th column and the $(i,j)$th element of the matrix $\bs{V}$ by $\bs{v}_i$ and $v_{ij}$, respectively.
	%     We show that the $i$th element of $\hat{\bs{p}}_{\bs{xy}}^{pg}$ is equal to that of $\hat{\bs{p}}_{\bs{xy}}^{cg}$. 
	%     \begin{equation*}
		%         (\hat{\bs{p}}_{\bs{xy}}^{cg})_i = \left(\bs{V}^H\left(\frac{1}{R}\sum\limits_{r=1}^{R} \bs{x}_r \bs{y}_r^H\right) \bs{V}\right)_{i,i} = \frac{1}{R}\sum\limits_{r=1}^{R} \bs{v}_i^H \bs{x}_r \bs{y}_r^H \bs{v}_i = \frac{1}{R}\sum\limits_{r=1}^{R}  (\bs{v}_i^H \bs{x}_r) (\bs{v}_i^H \bs{y}_r)^* = (\hat{\bs{p}}_{\bs{xy}}^{pg})_i
		%     \end{equation*}
	%     Therefore, $\hat{\bs{p}}_{\bs{xy}}^{pg} = \hat{\bs{p}}_{\bs{xy}}^{cg}$. 
	
	%     Now, we show that the $i$th element of $\hat{\bs{p}}_{\bs{xy}}^{cg}$ is equal to that of $\hat{\bs{p}}_{\bs{xy}}^{ls}$.
	%     It can be easily shown that $\bs{G}^H \bs{G} = I$. Thus, $\hat{\bs{p}}_{\bs{xy}}^{ls} = \bs{G}^H \hat{\bs{\sigma}}_{\bs{xy}}$. Then, 
	%     \begin{equation*}
		%         (\hat{\bs{p}}_{\bs{xy}}^{ls})_i = (\bs{G}^H \hat{\bs{\sigma}}_{\bs{xy}})_i = [v_{i1}\bs{v}_i^H, \ldots,v_{iN}\bs{v}_i^H]\hat{\bs{\sigma}}_{\bs{xy}} = \bs{v}_i^H \hat{\bs{\Sigma}}_{\bs{xy}} \bs{v}_i = (\hat{\bs{p}}_{\bs{xy}}^{cg})_i.
		%     \end{equation*}
	%     Thus, $\hat{\bs{p}}_{\bs{xy}}^{cg} = \hat{\bs{p}}_{\bs{xy}}^{ls}$. The proof is completed. 
	% \end{proof}
A proof is provided in Appendix. Hence, we refer to the three estimators as \textit{graph cross-periodogram} and denote them by $\hat{\bs{p}}_{XY}$.

%	We can evaluate the performance of the graph cross-periodogram by computing its MSE. Recall that the bias of the graph periodogram $\hat{\bs{p}}_{\bs{\mathrm{x}}}$ is zero, and its variance is $(2/R) \operatorname{diag}^2(\bs{p_{\mathrm{x}}})$ when  the process $\bs{\mathrm{x}}$ is Gaussian and $\bs{S}$ is symmetric. \Cref{prop:bias_var_graph_cross_periodogram} describes the bias and variance of the graph cross-periodogram $\hat{\bs{p}}_{\bs{\mathrm{xy}}}$.

\begin{proposition} \label{prop:bias_var_graph_cross_periodogram}
	For zero-mean jointly weakly stationary graph processes $X$ and $Y$ with GCSD $\bs{p}_{XY}$, the bias and the variance of the graph cross-periodogram $\hat{\bs{p}}_{XY}$ are computed as follows:
	\begin{itemize}
		\item[(a)] $E(\hat{\bs{p}}_{XY}) = \bs{p}_{XY}$. 
		\item[(b)] If $X$ and $Y$ are Gaussian and $\bs{S}$ is symmetric, $Var(\hat{\bs{p}}_{XY}) = \frac{1}{R} \left(\operatorname{diag}(\lvert \bs{p}_{XY} \rvert^2) + \operatorname{diag}(\bs{p}_{X})  \operatorname{diag}(\bs{p}_{Y})\right)$, where $\bs{p}_{X}$ and $\bs{p}_{Y}$ are the GPSDs of $X$ and $Y$, respectively.
	\end{itemize}
\end{proposition} 

A proof is given in Appendix.

\begin{remark}
	If $X=Y$, then the results in \Cref{prop:bias_var_graph_cross_periodogram} reduce to those of the periodogram described in \cite{Marques2017}.  
\end{remark}

\Cref{prop:bias_var_graph_cross_periodogram} implies that the MSE of the graph cross-periodogram decreases as the number of realizations $R$ increases. The right panel of \Cref{fig:cpsd_Rs} shows the effect of $R$ on the performance of the graph cross-periodogram evaluated using the Karate club network \cite{Zachary1977} in the left panel of \Cref{fig:cpsd_Rs}. To generate jointly weakly stationary graph processes, we utilize \Cref{def:jointweakstationary}(a). Specifically, we generate a white input $\bs{\epsilon}$ from $N(0,1)$ and then apply two linear shift-invariant filters, $\bs{\mathrm{H}}_{\text{Mex}}$ and $\bs{\mathrm{H}}_{\text{heat}}$, with frequency responses defined as $\tilde{h}_{\text{Mex}} (\lambda) = 5\lambda e^{-25\lambda^2}$ and $\tilde{h}_{\text{heat}} (\lambda) = e^{-10\lambda}$, respectively. Here, $\lambda$ denotes the normalized graph frequency in the interval $[0,1]$. For implementation, we use the Chebyshev polynomial approximation \cite{Phillips2006} to approximate these filters with a polynomial form $\bs{\mathrm{H}} = \sum_{\ell=1}^{N-1} h_\ell \bs{S}^{\ell}$, where the graph Laplacian serves as a graph shift operator $\bs{S}$ in the analysis. The true GCSD can be evaluated as $\tilde{h}_{\text{Mex}} (\lambda) \tilde{h}_{\text{heat}}^{*} (\lambda)$ by \Cref{prop:spectralconvolution_cpsd}.  As expected from \Cref{prop:bias_var_graph_cross_periodogram}, the accuracy of the graph cross-periodogram in estimating the true GCSD improves as $R$ increases. 
\begin{figure}
	\centering
	\subfloat
	{\includegraphics[width=0.48\linewidth]{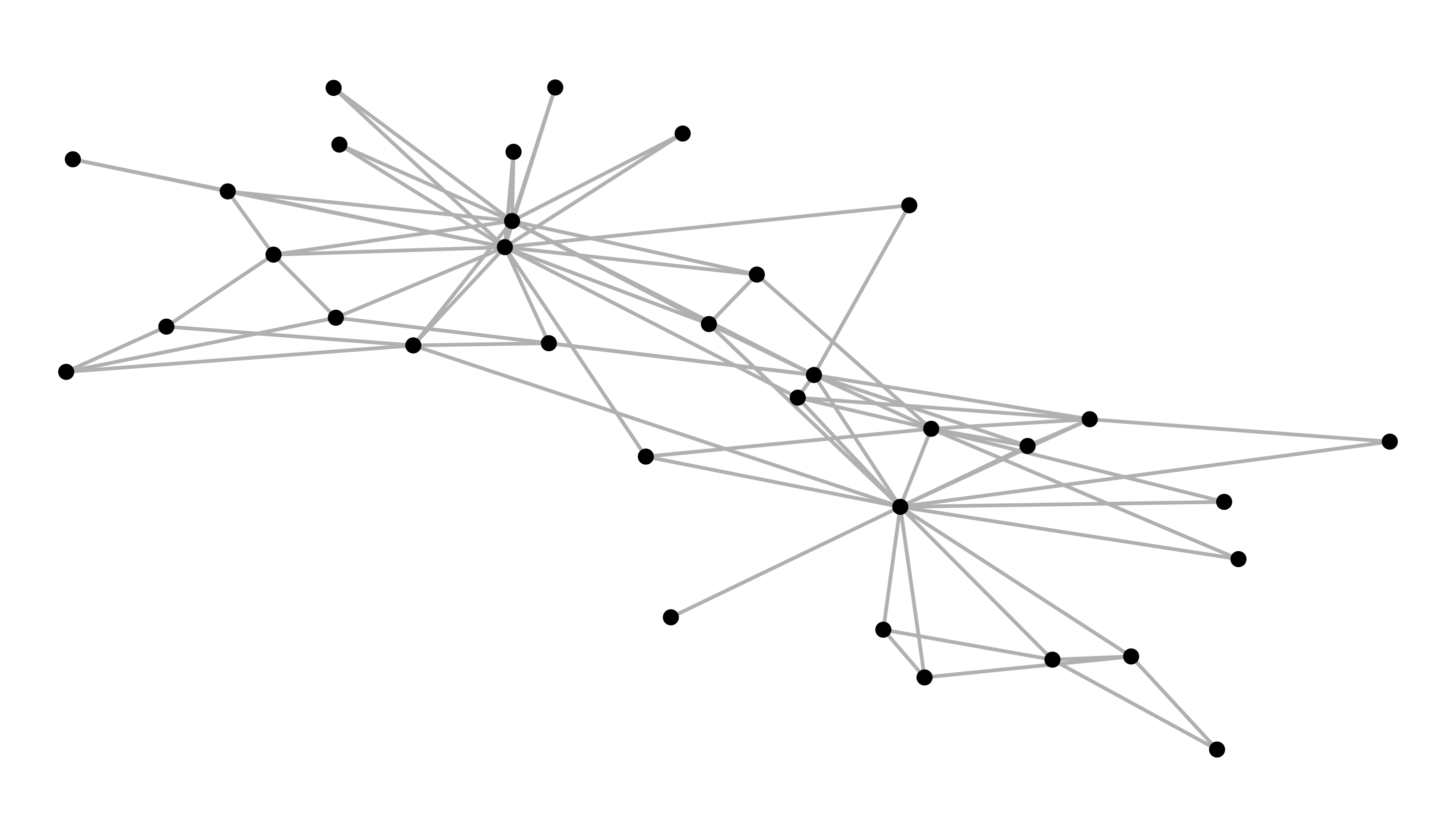}} 
	\hfill
	%\end{subfloat}
	\subfloat
	{\includegraphics[width=0.48\linewidth]{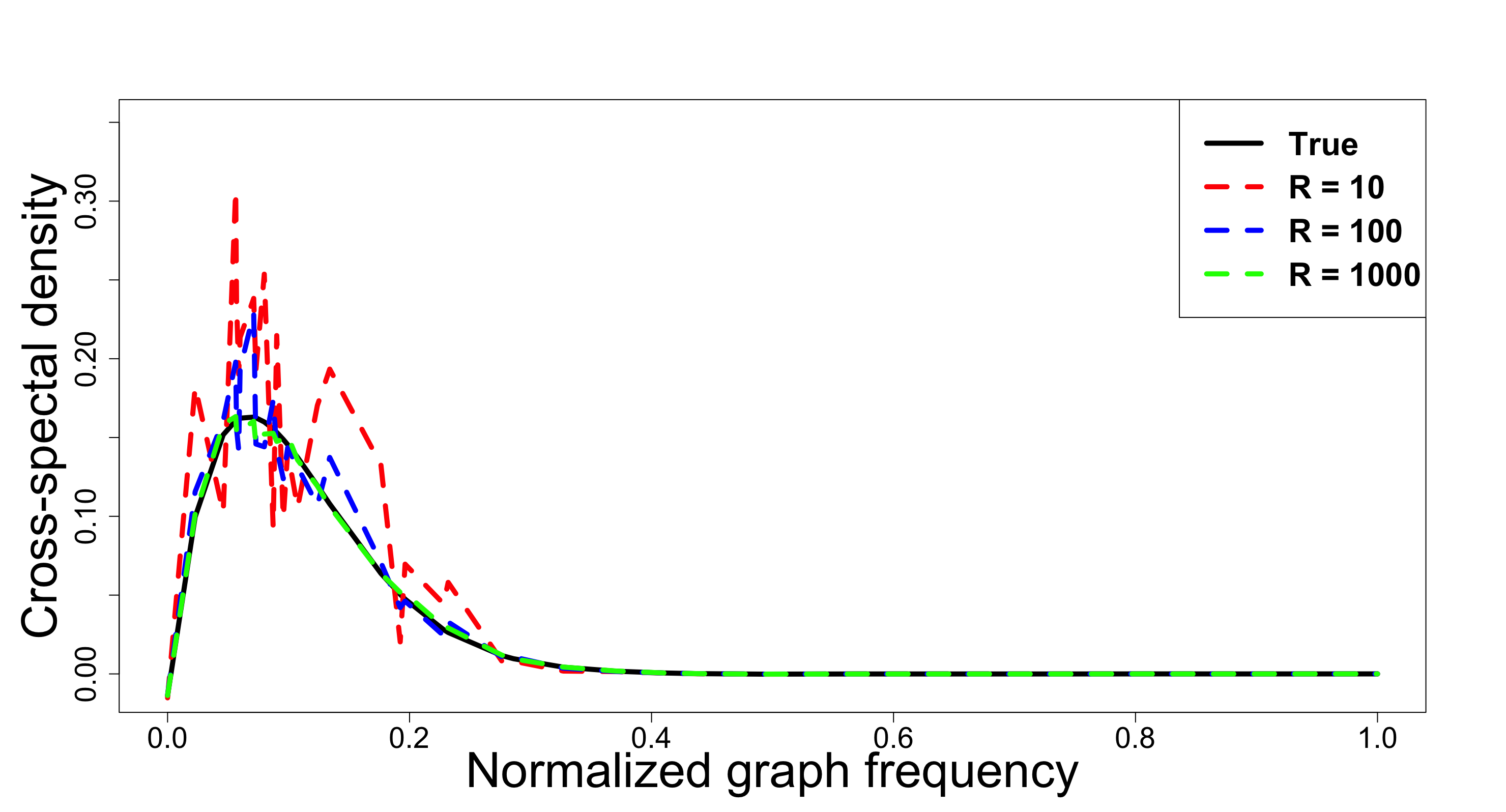}}
	%\end{subfloat}
	\caption{The Karate club network (left) and graph cross-periodograms for different values of $R$ on the Karate club network (right).}
	\label{fig:cpsd_Rs}
\end{figure}

%	\begin{figure} \label{fig:cpsd_estimation}
	%		\centering
	%		\includegraphics[width=0.99\textwidth]{figures/cpsd_RsMs.png}
	%		\caption{Cross-spectral density estimation results: graph cross-periodograms for different values of $R$ on the Karate club network (left), and windowed graph Fourier transform-based estimators for different values of $K$ on the Minnesota road network (right).}
	%		\label{fig:cpsd_RsMs}
	%	\end{figure}

\subsubsection{Windowed average graph cross-periodogram} \label{sec:windowavgcrossperiodogram}
As discussed in \Cref{sec:windowagvgperiodogram}, utilizing multiple realizations of graph processes to estimate GCSD is impractical, so a data window can be adopted. Given realizations $\bs{x}$ and $\bs{y}$ of the processes $X$ and $Y$, with a data window $\bs{w} \in \mathbb{R}^N$ such that $\lVert \bs{w} \rVert_2^2=N$, we define a windowed graph cross-periodogram as
\begin{align*}
	\hat{\bs{p}}_{XY}^{\bs{w}} = \llangle \bs{V}^H \bs{x_w}, \bs{V}^H \bs{y_w} \rrangle & = \llangle \bs{V}^H \operatorname{diag}(\bs{w})\bs{x}, \bs{V}^H \operatorname{diag}(\bs{w}) \bs{y} \rrangle \\
	& = \llangle \tilde{\bs{W}} \bs{V}^H \bs{x}, \tilde{\bs{W}} \bs{V}^H \bs{y} \rrangle,
\end{align*}
where  ~$\bs{x_w} = \operatorname{diag}(\bs{w})\bs{x}$~ and ~$\bs{y_w} = \operatorname{diag}(\bs{w})\bs{y}$~ denote windowed signals, and $\tilde{\bs{W}}=\bs{V}^H \operatorname{diag}(\bs{w}) \bs{V}$ is defined as the dual of the windowing operator $\operatorname{diag}(\bs{w})$ in the graph frequency domain.

The following proposition shows that the windowed graph cross-periodogram $\hat{\bs{p}}_{XY}^{\bs{w}}$ is biased unlike the graph cross-periodogram $\hat{\bs{p}}_{XY}$. 
\begin{proposition} \label{prop:windowcrossperiodogram_expectation}
	Let $\bs{p}_{XY}$ be the GCSD of zero-mean jointly weakly stationary graph processes $X$ and $Y$. The expectation of the windowed graph cross-periodogram is
	\begin{equation*}
		E(\hat{\bs{p}}_{XY}^{\bs{w}}) = (\tilde{\bs{W}} \circ \tilde{\bs{W}}^*)\bs{p}_{XY}.
	\end{equation*}
\end{proposition}
A proof is given in Appendix. 	
% \begin{proof}
	%     We can rewrite $\hat{\bs{p}}_{\bs{\mathrm{xy}}}^{\bs{w}} = \operatorname{diag}(\tilde{\bs{W}} \bs{V}^H \bs{x}  \bs{y}^H \bs{V} \tilde{\bs{W}}^H)$. Then, the following equation can be obtained by taking expectation.
	%     \begin{equation*}
		%         E(\hat{\bs{p}}_{\bs{\mathrm{xy}}}^{\bs{w}}) = \operatorname{diag}(\tilde{\bs{W}} \bs{V}^H \bs{\Sigma_{\mathrm{xy}}} \bs{V} \tilde{\bs{W}}^H) = \operatorname{diag}(\tilde{\bs{W}} \operatorname{diag}(\bs{p_{\mathrm{xy}}}) \tilde{\bs{W}}^H) = (\tilde{\bs{W}} \circ \tilde{\bs{W}}^*)\bs{p_{\mathrm{xy}}}.
		%     \end{equation*}
	%     The proof is completed.
	% \end{proof}
By utilizing multiple windows from a window bank $\mathcal{W} = \{\bs{w}_m\}_{m=1}^M$, we define a windowed average graph cross-periodogram as
\begin{equation*} \label{eq:windowavg_graphcross}
	\hat{\bs{p}}_{XY}^{\mathcal{W}} = \frac{1}{M} \sum\limits_{m=1}^{M} \hat{\bs{p}}_{XY}^{\bs{w}_m} = \frac{1}{M} \sum\limits_{m=1}^{M} \llangle \tilde{\bs{W}}_m \bs{V}^H \bs{x}, \tilde{\bs{W}}_m \bs{V}^H \bs{y} \rrangle,
\end{equation*}
where  $\tilde{\bs{W}}_m=\bs{V}^H \operatorname{diag}(\bs{w}_m) \bs{V}$. %We calculate the bias and variance of the windowed average graph cross-periodogram $\hat{\bs{p}}_{\bs{\mathrm{xy}}}^{\mathcal{W}}$, and the results are given in \Cref{prop:bias_var_window_graph_cross}. The proof is provided in Appendix \ref{appendix:proof}.

\begin{proposition} \label{prop:bias_var_window_graph_cross}
	Let $X$ and $Y$ be zero-mean jointly weakly stationary graph processes $X$ and $Y$ with GCSD $\bs{p}_{XY}$, and $\mathcal{W} = \{\bs{w}_m\}_{m=1}^M$. Then, the bias and the variance of the windowed average graph cross-periodogram $\hat{\bs{p}}_{XY}^{\mathcal{W}}$ are as follows: 
	\begin{itemize}
		\item[(a)] $E(\hat{\bs{p}}_{XY}^{\mathcal{W}}) = \frac{1}{M} \sum_{m=1}^{M} (\tilde{\bs{W}}_m \circ \tilde{\bs{W}}_{m}^*)\bs{p}_{XY}.$ %i.e., the bias is $\frac{1}{M} \sum_{m=1}^{M} (\tilde{\bs{W}}_m \circ \tilde{\bs{W}}_{m}^*)\bs{p_{\mathrm{xy}}} - \bs{p_{\mathrm{xy}}}$. 
		\item[(b)] If $X$ and $Y$ are Gaussian and $\bs{S}$ is symmetric,  
		\begin{alignat*}{2}
			\operatorname{tr}&(Var(\hat{\bs{p}}_{XY}^{\mathcal{W}})) \\
			&= \frac{1}{M^2} \sum\limits_{m=1}^{M} \sum\limits_{m'=1}^{M} \operatorname{tr}&&\big((\tilde{\bs{W}}_{m,m'}\bs{p}_{XY})(\tilde{\bs{W}}_{m,m'}\bs{p}_{XY})^H \\ & && + (\tilde{\bs{W}}_{m,m'}\bs{p}_{X})(\tilde{\bs{W}}_{m,m'}\bs{p}_{Y})^H \big),
		\end{alignat*}
		where $\tilde{\bs{W}}_{m,m'} = \tilde{\bs{W}}_{m} \circ \tilde{\bs{W}}_{m'}^*$.
	\end{itemize}
\end{proposition}
A proof is provided in Appendix.	

\begin{remark}
	If $X=Y$, the results in \Cref{prop:bias_var_window_graph_cross} reduce to those of the windowed average periodogram in \cite{Marques2017}.  
\end{remark}

\begin{proposition} \label{prop:gcpsd_windowdesign}
	Consider two windows $\bs{w}_m$ and $\bs{w}_{m'}$ and assume that the distance between any node in the support of $\bs{w}_m$ and any node in the support of $\bs{w}_{m'}$ is larger than $(L_1 + L_2)$ hops. If two processes $X$ and $Y$ are responses of graph filters $\bs{\mathrm{H}}_1$ and $\bs{\mathrm{H}}_2$ of degrees $L_1$ and $L_2$ to white inputs $\bs{\epsilon}_1$ and $\bs{\epsilon}_2$ such that $Cov(\bs{\epsilon}_1, \bs{\epsilon}_2) = \bs{I}$, then it holds that $\operatorname{tr}( (\tilde{\bs{W}}_{m,m'} \bs{p}_{XY})(\tilde{\bs{W}}_{m,m'} \bs{p}_{XY})^H) = 0$, where $\tilde{\bs{W}}_{m,m'} = \tilde{\bs{W}}_{m} \circ \tilde{\bs{W}}_{m'}^*$.
\end{proposition}
A proof is provided in Appendix.

\begin{remark} \label{rmk:prop_windowdesign}
	\Cref{prop:gcpsd_windowdesign} plays an important role in window design. The MSE of $\hat{\bs{p}}_{XY}^{\mathcal{W}}$ is equal to the sum of squared bias and $\operatorname{tr}(Var(\hat{\bs{p}}_{XY}^{\mathcal{W}}))$. From \Cref{prop:bias_var_window_graph_cross}(b), $\operatorname{tr}(Var(\hat{\bs{p}}_{XY}^{\mathcal{W}}))$ can be decomposed into two terms, one with $m=m'$ and the other with $m \neq m'$. The former term is  $\frac{1}{M^2}\sum_{m=1}^{M} \operatorname{tr}\big((\tilde{\bs{W}}_{m,m}\bs{p}_{XY})(\tilde{\bs{W}}_{m,m}\bs{p}_{XY})^H + (\tilde{\bs{W}}_{m,m}\bs{p}_{X})(\tilde{\bs{W}}_{m,m}\bs{p}_{Y})^H\big)$. By \Cref{prop:windowcrossperiodogram_expectation}, it follows that  $\operatorname{tr}\big( (\tilde{\bs{W}}_{m,m} \bs{p}_{XY})(\tilde{\bs{W}}_{m,m} \bs{p}_{XY})^H\big) = \lVert E(\hat{\bs{p}}_{XY}^{\bs{w}_m}) \rVert_2^2$. Given that $\lVert \bs{w}_m \rVert_2^2 = N$, it can be approximated by $\lVert \bs{p}_{XY} \rVert_2^2$. Similarly, we have $\operatorname{tr}( (\tilde{\bs{W}}_{m,m} \bs{p}_{X})(\tilde{\bs{W}}_{m,m} \bs{p}_{Y})^H) =  \langle E(\hat{\bs{p}}_{X}^{\bs{w}_m}), E(\hat{\bs{p}}_{Y}^{\bs{w}_m}) \rangle$, which can be approximated by $\langle \bs{p}_{X}, \bs{p}_{Y} \rangle$. Thus, the former term is close to $\frac{1}{M}(\lVert \bs{p}_{XY} \rVert_2^2 + \langle \bs{p}_{X}, \bs{p}_{Y} \rangle) = \frac{1}{M} \operatorname{tr} \left(\operatorname{diag}(\lvert \bs{p}_{XY} \rvert^2) + \operatorname{diag}(\bs{p}_{X})  \operatorname{diag}(\bs{p}_{Y})\right)$. This implies that the latter term (when $m \neq m'$) contains the effect of the correlation between $M$ windowed bivariate signals, referring to \Cref{prop:bias_var_graph_cross_periodogram}(b). \Cref{prop:crosscov_element} indicates that certain jointly weakly stationary graph processes exhibit localized cross-covariance structures. In this case,  \Cref{prop:gcpsd_windowdesign} suggests that designing non-overlapping windows can mitigate the effects of correlation between the windowed bivariate signals by reducing the latter term in the decomposition of $Var(\hat{\bs{p}}_{XY}^{\mathcal{W}})$. This leads to a reduction in the MSE of $\hat{\bs{p}}_{XY}^{\mathcal{W}}$. Moreover, according to Proposition 5 in \cite{Marques2017}, if  the distance between any node in the support of $\bs{w}_m$ and any node in the support of $\bs{w}_{m'}$ is larger than $\max(L_1, L_2)$ hops, $\tilde{\bs{W}}_{m,m'}\bs{p}_{X} = \tilde{\bs{W}}_{m,m'}\bs{p}_{Y}=\bs{0}$. Therefore, if the distance between any node in support of $\bs{w}_m$ and any node in support of $\bs{w}_{m'}$ is larger than  $\max(2L_1, 2L_2)$ hops, the latter term in the decomposition of $\operatorname{tr}(Var(\hat{\bs{p}}_{XY}^{\mathcal{W}}))$ becomes zero, which means that $Var(\hat{\bs{p}}_{XY}^{\mathcal{W}})$ behaves as if the windowed signals are independent.
\end{remark}

\subsection{Graph coherence}
From the definitions of graph power spectral density and graph cross-spectral density, we define the graph coherence (or magnitude-squared graph coherence) between two jointly weakly stationary processes $X$ and $Y$, denoted by $\bs{c}_{XY}$.

\begin{definition} \label{def:graph_coherence}
	The coherence between zero-mean jointly weakly stationary graph processes $X$ and $Y$ is defined as $\bs{c}_{XY} = \frac{\lvert \bs{p}_{XY} \rvert^2}{\bs{p}_{X} \bs{p}_{Y}}$, where the product and division operators are applied elementwise.
\end{definition}

%	The graph coherence satisfies the following inequality similar to coherence in classical signal processing, which is proved in Appendix \ref{appendix:proof}.	
\begin{proposition} \label{prop:coherence_ineq}
	The graph coherence satisfies the following inequality,
	$
	\bs{0} \le \bs{c}_{XY} \le \bs{1},
	$
	where $\le$ indicates elementwise inequality. The equality $\bs{c}_{XY}=\bs{0}$ holds when $X$ and $Y$ are uncorrelated, and $\bs{c}_{XY}=\bs{1}$ holds when $Y = \bs{\mathrm{H}} X$ for a filter $\bs{\mathrm{H}}$.
\end{proposition}
A proof is given in Appendix.

\section{Practical implementation} \label{sec:practical}
\subsection{Window design} \label{sec:windowdesign}
Designing windows is crucial for the estimation performance of the windowed average graph cross-periodogram. Marques et al. \cite{Marques2017} addressed window design for the windowed average graph periodogram, but this remains a challenging and open problem. As suggested by \cite{Marques2017}, for each window $\bs{w}_m$, $\tilde{\bs{W}}_{m}$ should be close to the identity matrix so that the windowed signal and the original signal are similar in the graph frequency domain. In this regard, we can adopt random windows that are close to the identity matrix. Also, as mentioned in \Cref{rmk:prop_windowdesign}, if a community structure exists in the graph, we can design non-overlapping local windows corresponding to each community. 

In this study, we use random windows for practical implementation in the numerical experiments in \Cref{sec:numerical}. Specifically, we generate $\tilde{\bs{W}}_{m}$ as an identity matrix with random Gaussian noise added for each element. Then, we obtain $\bs{w}_m$ from $\operatorname{diag}(\bs{V} \tilde{\bs{W}}_{m} \bs{V}^H)$. However, it is important to note that this is not an optimal design and tends to degrade performance, especially when the number of nodes is small. Thus, we consider an alternative estimator for implementation. 

\subsection{Alternative estimator for graph cross-spectral density} \label{sec:alternative}
Perraudin and Vandergheynst \cite{Perraudin2017} provided a perspective on the PSD estimation of Bartlett's and Welch's methods in classical signal processing, interpreting them as the averaging over time of the squared coefficients derived from a windowed Fourier transform.  Based on this insight, they proposed an estimator of GPSD using the windowed graph Fourier transform introduced by \cite{Shuman2016}. Specifically, for a graph kernel $g$ that is concentrated at the origin (e.g., $g(\lambda) = e^{-\lambda^2 / \sigma^2}$ for a constant $\sigma^2$), $g_k$ is defined by shifting $g$ as
\begin{equation} \label{eq:gk}
	g_k(\lambda_\ell) = g(\lambda_\ell - k \tau), \quad k=1,\ldots, K, \; \tau = \frac{\lambda_{max}}{K},
\end{equation}
where $\lambda_\ell$ ($1 \le \ell \le N)$ are the eigenvalues of $\bs{S}$ and $\lambda_{max}$ is the maximum value among them. Then, the windowed graph Fourier transform-based estimator is defined as
\begin{equation} \label{eq:WFT_periodogram}
	\hat{\bs{p}}_{X}^{WF}(k \tau) = \frac{\lVert g_k(\bs{S}) \bs{x} \rVert_2^2}{\lVert g_k(\bs{S})\rVert_F^2} = \frac{\sum_{n=1}^{N} C_{n,k}^2}{\lVert g_k(\bs{S})\rVert_F^2},
\end{equation}
where $\lVert \cdot \rVert_F$ denotes the Frobenius norm, $g_k(\bs{S}) = \bs{V} \operatorname{diag}(g_k(\lambda_\ell)) \bs{V}^H$, and $C_{n,k}$ denote the coefficients of the windowed graph Fourier transform, defined as
\begin{equation*}
	C_{n,k} = \langle \bs{x}, \mathcal{T}_n g_k \rangle = (g_k(\bs{S}) \bs{x})_n.
\end{equation*}
Here, $\mathcal{T}_n$ represents the localization operator such that $\mathcal{T}_n g = \sum_{\ell=1}^{N} g(\lambda_\ell) (\bs{v}_\ell^*)_n \bs{v}_\ell$, and  $(\cdot)_n$ denotes the $n$th element operation on a vector. %This approach implicitly assumes that GPSD has a functional form continuously defined in the graph frequency domain.

Along the same line with the windowed graph Fourier transform-based estimator of GPSD in (\ref{eq:WFT_periodogram}), we propose a windowed graph Fourier transform-based estimator of GCSD as
\begin{equation*} \label{eq:WFT_cross_periodogram}
	\hat{\bs{p}}^{WF}_{XY}(k\tau) = \frac{\langle g_k(\bs{S}) \bs{x}, g_k(\bs{S}) \bs{y} \rangle}{\lVert g_k(\bs{S}) 
		\rVert_F^2} = \frac{\sum_{n=1}^{N} \langle C^{(1)}_{n,k}, C^{(2)}_{n,k} \rangle}{\lVert g_k(\bs{S}) 
		\rVert_F^2}
\end{equation*}
for $k=1,\ldots, K$ and $\tau = \lambda_{max}/K$, where $C^{(1)}_{n,k} = \langle \bs{x}, \mathcal{T}_n g_k \rangle = (g_k(\bs{S}) \bs{x})_n$ and $C^{(2)}_{n,k} = \langle \bs{y}, \mathcal{T}_n g_k \rangle = (g_k(\bs{S}) \bs{y})_n$ for the localization operator $\mathcal{T}_n$ and $g_k$ defined in (\ref{eq:gk}). 

We remark that as noted  in \cite{Perraudin2017}, the difference between $\hat{\bs{p}}^{WF}_{XY}$ and the true GCSD $\bs{p}_{XY}$ is insignificant when $\sigma \gg \lambda_{max} / N$. However, as $\sigma$ decreases, the difference may increase.

We now investigate how the number of filters $K$ affects the performance of the windowed graph Fourier transform-based estimator on the Minnesota road network, comprising 2642 nodes and 6606 edges, all with weights equal to 1, which is shown in the top left panel of \Cref{fig:minnesota_cross_spectrum}. Jointly weakly stationary graph processes are generated using $\bs{\mathrm{H}}_{\text{ds}}$ and $\bs{\mathrm{H}}_{\text{high}}$ with frequency responses defined as $\tilde{h}_{\text{ds}} (\lambda) = \sin(15\lambda)e^{-5\lambda}$ and $\tilde{h}_{\text{high}} (\lambda) = \lambda e^{-\lambda}$, respectively, resulting in the true GCSD being $\tilde{h}_{\text{ds}} (\lambda) \tilde{h}_{\text{high}}^{*} (\lambda)$. As in \Cref{sec:graph_crossperiodogram}, we use the Chebyshev polynomial approximation and the graph Laplacian as a graph shift operator $\bs{S}$. Following \cite{Perraudin2017}, the graph kernel $g$ is defined as $g(\lambda) = e^{-\lambda^2 / \sigma^2}$, where $\sigma^2$ is chosen as $(K+1)\lambda_{max}/K^2$. The estimation results for the windowed graph Fourier transform-based estimators for different $K$'s are shown in \Cref{fig:cpsd_Ms}. As $K$ increases, estimation performance improves. However, if $K$ is too large, estimation becomes wiggly, and performance degrades. This is consistent with the discussion above about the effect of $\sigma$ on estimation performance. 

\begin{figure}[!t]
	\centering
	\includegraphics[width=0.6\textwidth]{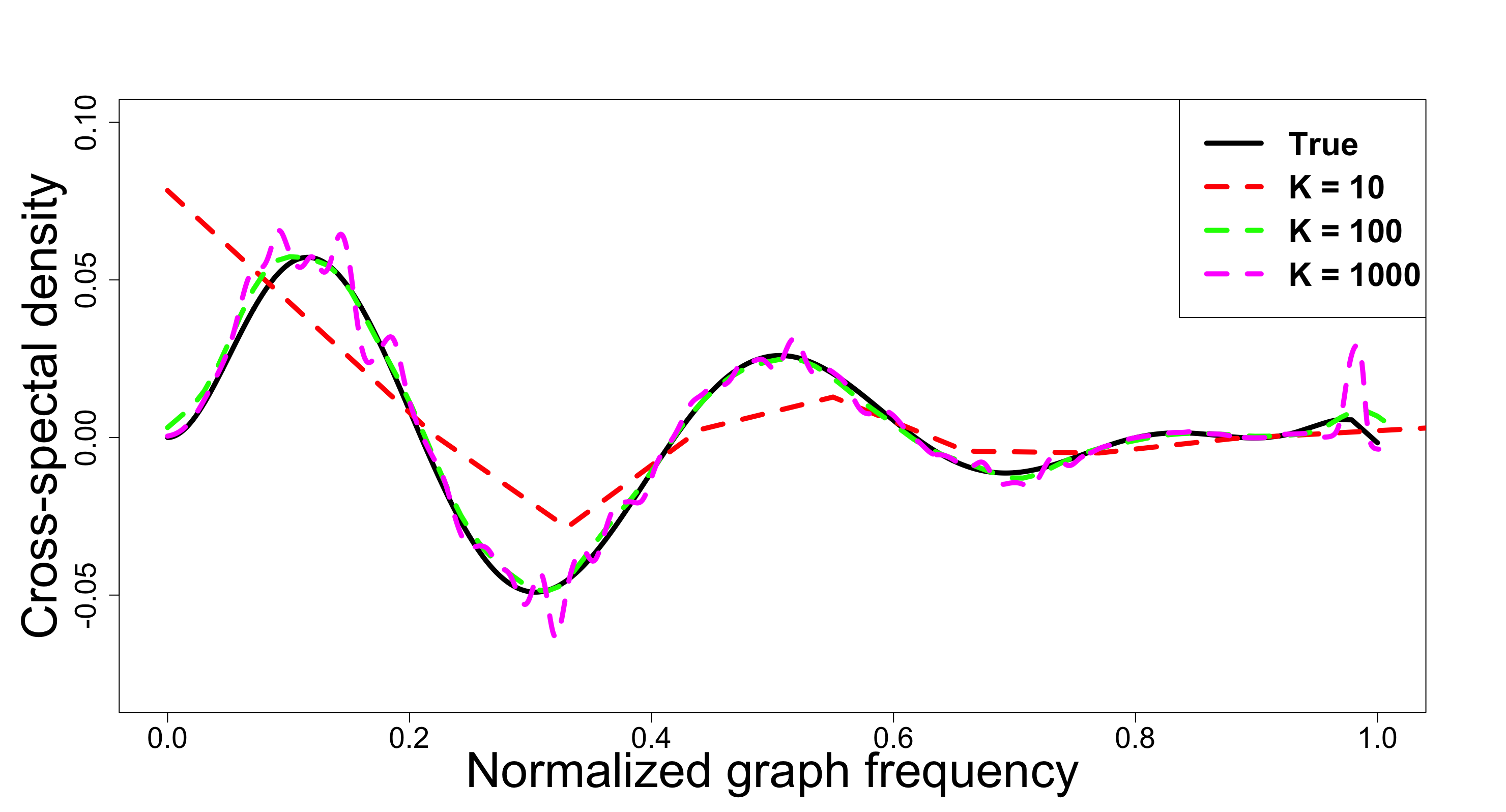}
	\caption{The windowed graph Fourier transform-based estimates for different values of $K$ on the Minnesota road network.}
	\label{fig:cpsd_Ms}
\end{figure}

%\begin{figure} \label{fig:graph_filters}
%    \centering
%    \includegraphics[width=0.99\textwidth]{figures/graph_filters.png}
%    \caption{Frequency responses of each graph filter.}
%    \label{fig:graph_filters}
%\end{figure}

% \section{Numerical experiments} \label{sec:numerical}
\section{Numerical experiments} \label{sec:numerical}
\subsection{Detection of shared components} \label{sec:sharedcomponent}
In this section, we demonstrate the importance of cross-spectral analysis on graphs by investigating the graph cross-periodogram for different irregular graphs. In particular, we generate two graph signals that share the same frequency, which is indistinguishable from the graph periodogram of each signal. However, in the cross-spectral analysis, this frequency can be revealed by the graph cross-periodogram. 	

We consider three irregular graphs: the Seoul Metropolitan railroad network, a random sensor network, and the Minnesota road network. In all scenarios, we compute the estimators of GPSD and GCSD using both windowed averaging and windowed graph Fourier-based methods. Here, the graph Laplacian serves as a graph shift operator. For the window averaging approach, we utilize random windows as described in \Cref{sec:windowdesign}, appropriately choosing the number of windows $M$. For the windowed graph Fourier-based approach, we use the same graph kernel $g$ in \Cref{sec:alternative}. The number of filters $K$ is chosen to ensure that $\sigma$ is significantly greater than $\lambda_{max} / N$, while still being large enough considering the number of nodes.  

We first consider the Seoul railroad network studied by \cite{Kim2023}. This network has 243 stations as nodes and is connected by railroads. The edge weight between the $i$th node and the $j$th node is determined by $A_{i,j} = \exp(-d^2(i,j) / ave^2)$, where $d(i,j)$ represents the distance between the $i$th and $j$th nodes computed by summing the lengths of the railroads connecting the two nodes, and $ave$ denotes the average distance between the nodes. We generate two signals  as $\bs{x}_{\text{metro}} = 5\bs{v}_{100}^{\text{metro}}+100\bs{v}_{50}^{\text{metro}}$ and $ \bs{y}_{\text{metro}} = 5\bs{v}_{100}^{\text{metro}}+100\bs{v}_{150}^{\text{metro}}$, where $\bs{v}_{i}^{\text{metro}}$ denotes the $i$th eigenvector of the graph Laplacian, corresponding to the $i$th eigenvalue $\lambda_i^{\text{metro}}$. 
\begin{figure}[!t]
	\centering
	\includegraphics[width=0.95\textwidth]{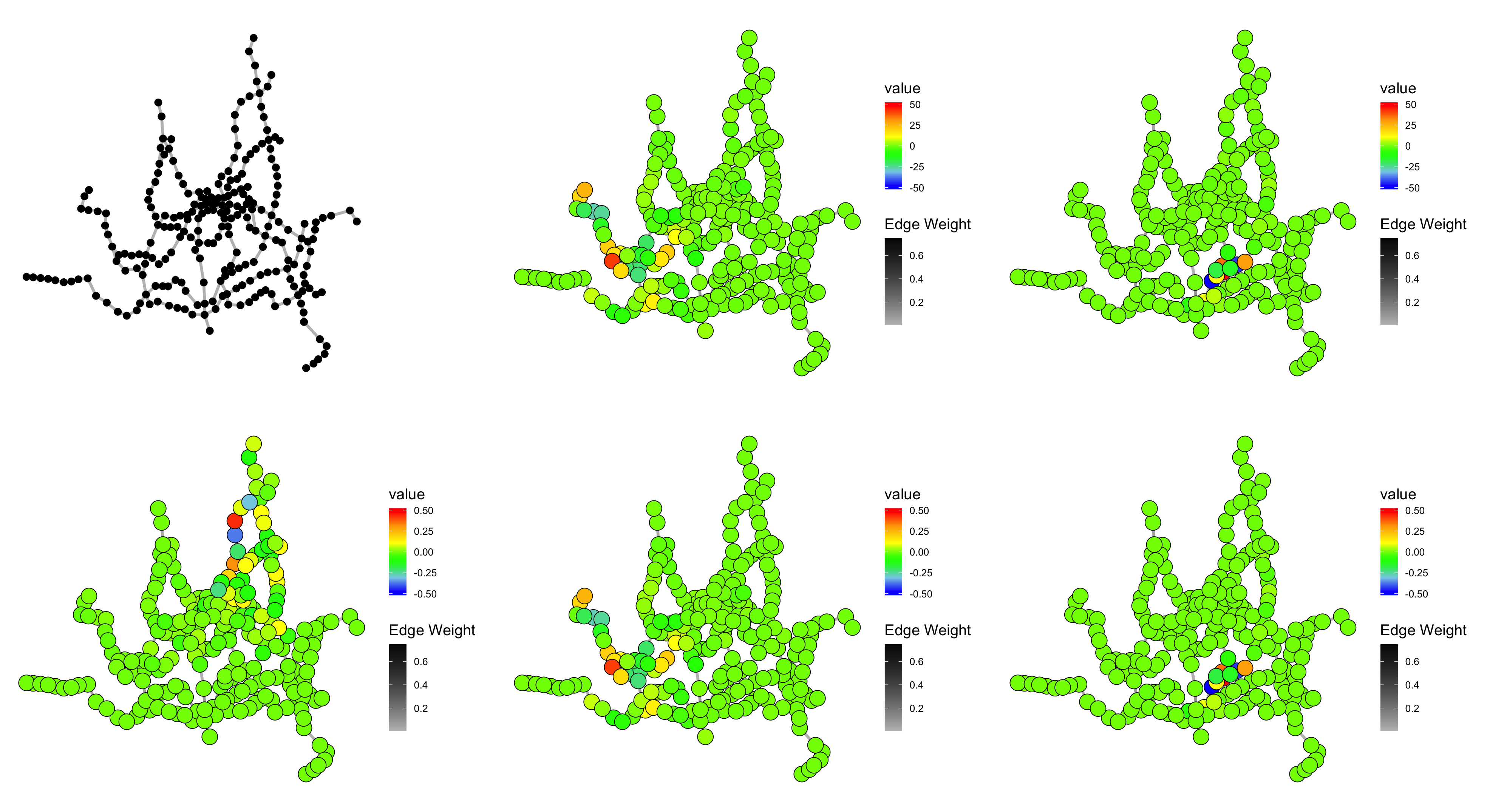}
	\caption{Visualization of the Seoul Metropolitan railroad network and signals: Seoul Metropolitan railroad network (top left), $\bs{x}_{\text{metro}}$ (top middle), $\bs{y}_{\text{metro}}$ (top right), $\bs{v}_{100}^{\text{metro}}$ (bottom left), $\bs{v}_{50}^{\text{metro}}$ (bottom middle), $\bs{v}_{150}^{\text{metro}}$ (bottom right).}
	\label{fig:metro_cross_spectrum}
\end{figure}
\begin{figure}[!t]
	\centering
	\includegraphics[width=0.95\textwidth]{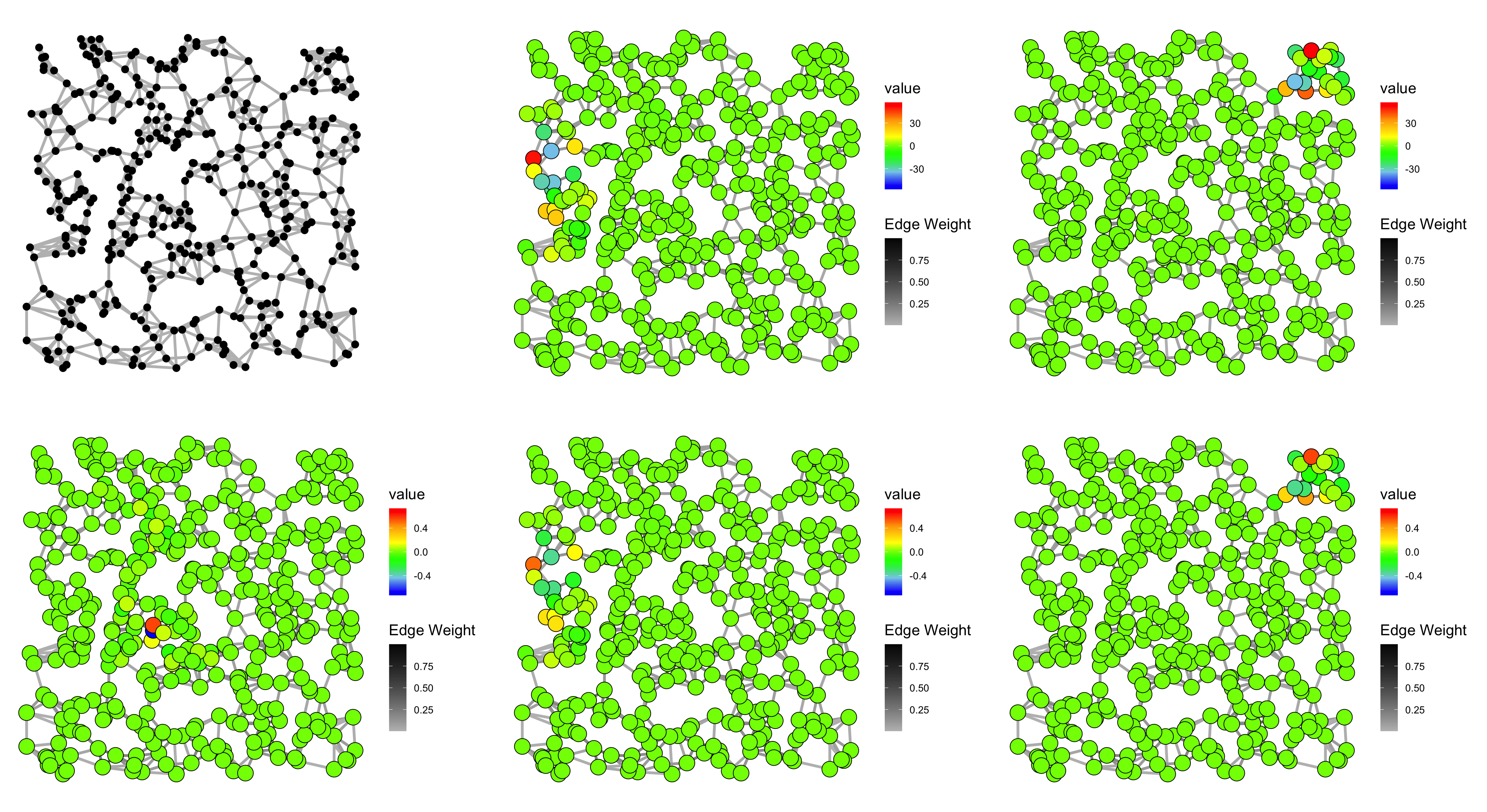}
	\caption{Visualization of a random sensor network and signals: random sensor network (top left), $\bs{x}_{\text{rsn}}$ (top middle), $\bs{y}_{\text{rsn}}$ (top right), $\bs{v}_{300}^{\text{rsn}}$ (bottom left), $\bs{v}_{100}^{\text{rsn}}$ (bottom middle), $\bs{v}_{200}^{\text{rsn}}$ (bottom right).}
	\label{fig:irregular_cross_spectrum}
\end{figure}
\begin{figure}[!t]
	\centering
	\includegraphics[width=0.95\textwidth]{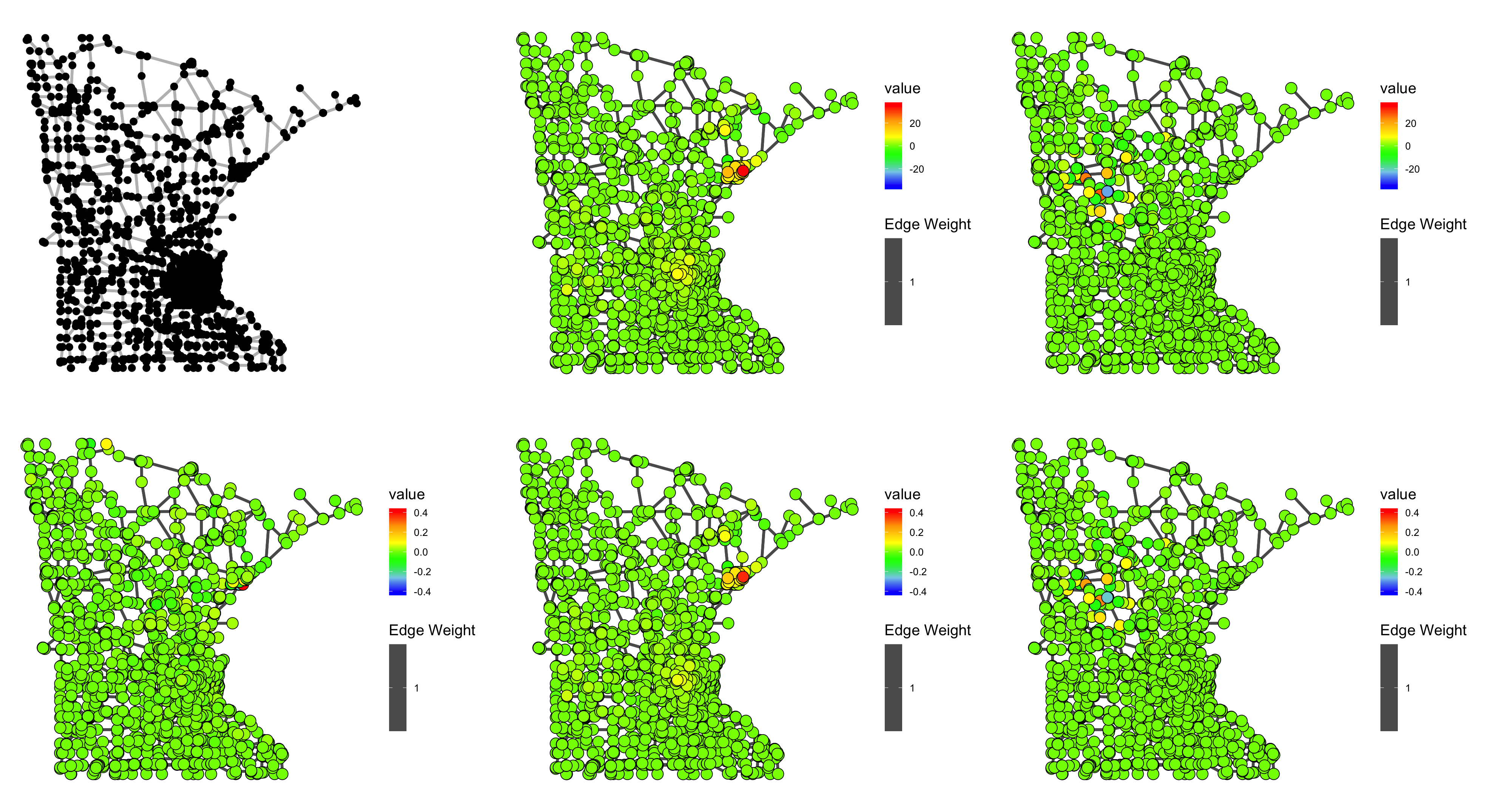}
	\caption{Visualization of the Minnesota road network and signals: Minnesota road network (top left), $\bs{x}_{\text{minne}}$ (top middle), $\bs{y}_{\text{minne}}$ (top right), $\bs{v}_{500}^{\text{minne}}$ (bottom left), $\bs{v}_{1000}^{\text{minne}}$ (bottom middle), $\bs{v}_{2000}^{\text{minne}}$ (bottom right).}
	\label{fig:minnesota_cross_spectrum}
\end{figure}
\Cref{fig:metro_cross_spectrum} shows that $\bs{x}_{\text{metro}}$ predominantly contains $\bs{v}_{50}^{\text{metro}}$, while $\bs{y}_{\text{metro}}$ primarily consists of $\bs{v}_{150}^{\text{metro}}$. This reasoning is consistent with what we observed in the graph periodograms of the two signals, as shown in Figures \ref{fig:cpsd_simul_wavg} and \ref{fig:cpsd_simul_wf}. For each estimator, the values of $M$ and $K$ are set to 100, respectively. However, based on these results, we cannot discern the shared component $\bs{v}_{100}^{\text{metro}}$ between the two signals; this identification becomes feasible through cross-spectral analysis. The GPSD estimates in Figures \ref{fig:cpsd_simul_wavg} and \ref{fig:cpsd_simul_wf} show peaks at the graph frequency corresponding roughly to $\lambda_{100}^{\text{metro}}$.  These observations indicate that the two signals share the component $\bs{v}_{100}^{\text{metro}}$.

Next, we conduct a similar analysis on a random sensor network shown in the top left panel of \Cref{fig:irregular_cross_spectrum}. The random sensor network is constructed by randomly selecting 400 nodes located in  $[0,20] \times [0,20]$, where the $x$- and $y$-coordinates are independently generated from a uniform distribution $\mathcal{U}(0,20)$. Each node is then connected to its five nearest neighbors. The edge weights are determined in the same way as the Seoul Metropolitan railroad network, where the distance represents the Euclidean distance. Two signals are generated as $\bs{x}_{\text{rsn}} = 5\bs{v}_{300}^{\text{rsn}}+100\bs{v}_{100}^{\text{rsn}}$ and $ \bs{y}_{\text{rsn}} = 5\bs{v}_{300}^{\text{rsn}}+100\bs{v}_{200}^{\text{rsn}}$, where $\bs{v}_{i}^{\text{rsn}}$ denotes the $i$th eigenvector of the graph Laplacian, corresponding to the $i$th eigenvalue $\lambda_i^{\text{rsn}}$. The common component between the two signals, $\bs{v}_{300}^{\text{rsn}}$, is uncovered in the GCSD estimates with $M=100$ and $K=50$, as shown in Figures \ref{fig:cpsd_simul_wavg} and \ref{fig:cpsd_simul_wf}.

\begin{figure}[!t]
	\centering
	\includegraphics[width=0.95\textwidth]{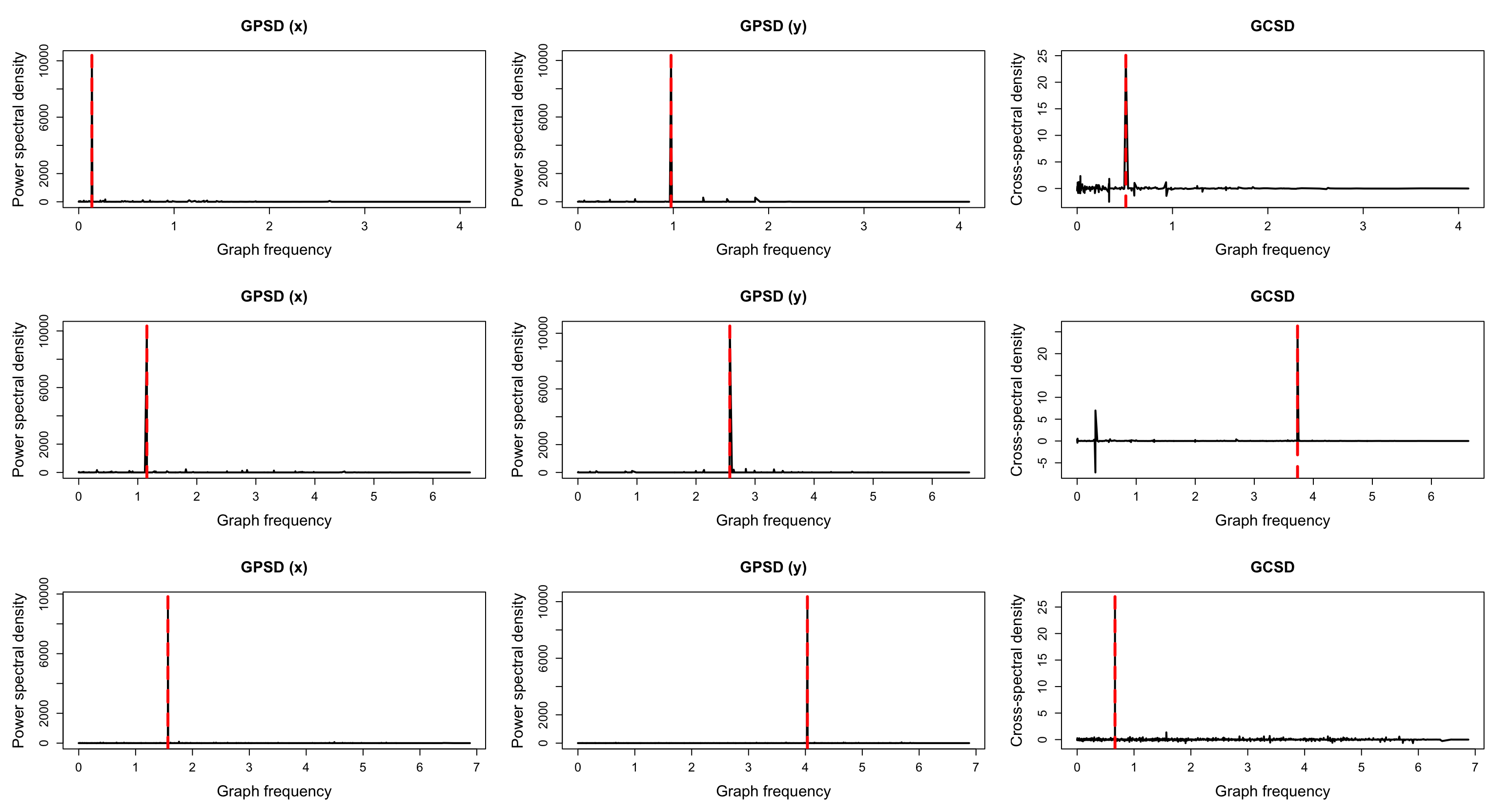}
	\caption{Cross-spectral analysis results on three networks: Seoul Metropolitan railroad network (top row), random sensor network (middle row), and Minnesota road network (bottom row). Each column, from left to right, illustrates the windowed average estimates of GPSD of $X$, $Y$, and GCSD, respectively.}
	\label{fig:cpsd_simul_wavg}
\end{figure}

\begin{figure}[!t]
	%		\captionsetup{font={stretch=1.1}}
	\centering
	\includegraphics[width=0.95\textwidth]{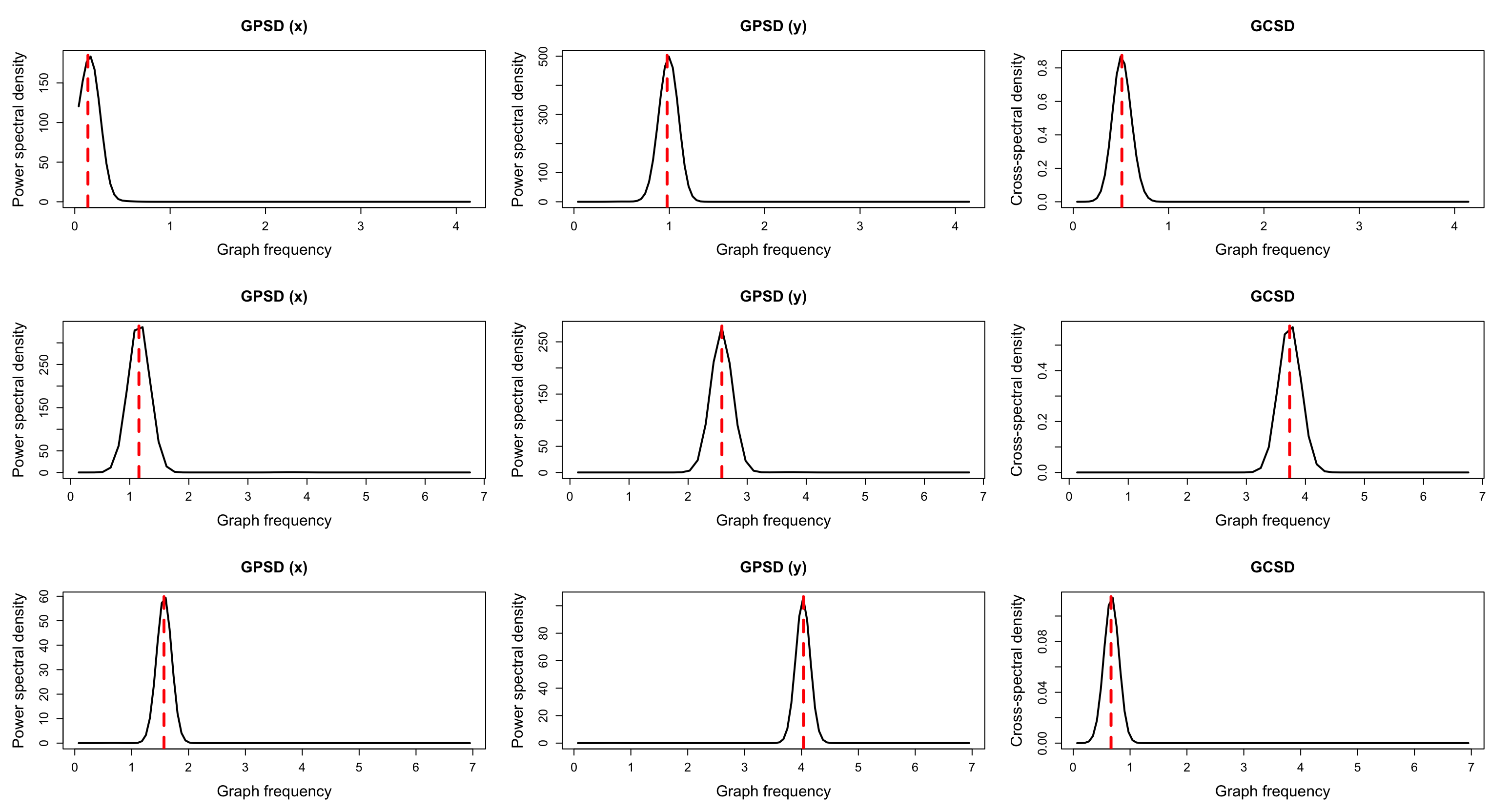}
	\caption{Cross-spectral analysis results on three networks: Seoul Metropolitan railroad network (top row), random sensor network (middle row), and Minnesota road network (bottom row). Each column, from left to right, illustrates the windowed graph Fourier transform-based estimates of GPSD of $X$, $Y$, and GCSD, respectively.}
	\label{fig:cpsd_simul_wf}
\end{figure}

Finally, we use the Minnesota road network for analysis. We generate two signals as $\bs{x}_{\text{minne}} = 5\bs{v}_{500}^{\text{minne}}+100\bs{v}_{1000}^{\text{minne}}$ and $ \bs{y}_{\text{minne}} = 5\bs{v}_{500}^{\text{minne}}+100\bs{v}_{2000}^{\text{minne}}$. Here, $\bs{v}_{i}^{\text{minne}}$ denotes the $i$th eigenvector of the graph Laplacian, corresponding to the $i$th eigenvalue $\lambda_i^{\text{minne}}$. The GCSD estimates obtained with $M=100$ and  $K=100$ reveal the common component between the two signals, $\bs{v}_{500}^{\text{minne}}$, as shown in Figures \ref{fig:cpsd_simul_wavg} and \ref{fig:cpsd_simul_wf}.

\subsection{Cross-spectral analysis on meteorological dataset}
We apply the proposed cross-spectral analysis method to a meteorological dataset collected hourly during January 2014 in Brest, France. The dataset contains temperature and humidity measurements recorded at 22 weather stations. The original data are available at \url{https://donneespubliques.meteofrance.fr}. To construct the graph, each node is connected to its six nearest neighbors, and the edge weights are determined in the same manner as the random sensor network used in \Cref{sec:sharedcomponent}. When calculating the distance between two nodes, we use the weighted 3D coordinates of each station as suggested by \cite{Perraudin2017}, where we multiply the $z$-coordinate by 5 since the elevation is considered more important than other coordinates. The resultant graph is shown in \Cref{fig:france_meteo}. We average the hourly data to obtain the monthly average temperature and humidity for January. Subsequently, we subtract the mean signals across nodes to ensure that they are observations from zero-mean processes. For the graph shift operator, we use the graph Laplacian whose $i$th eigenvector corresponding to the $i$th eigenvalue $\lambda_i^{\text{brest}}$ is denoted by $\bs{v}_{i}^{\text{brest}}$. Since the graph has few nodes, we estimate GPSD and GCSD using the windowed graph Fourier transform-based estimator with the same graph kernel $g$ as in the previous section. The number of filters $K$ is set to 40.
\begin{figure}[!t]
	\centering
	\includegraphics[width=0.95\textwidth]{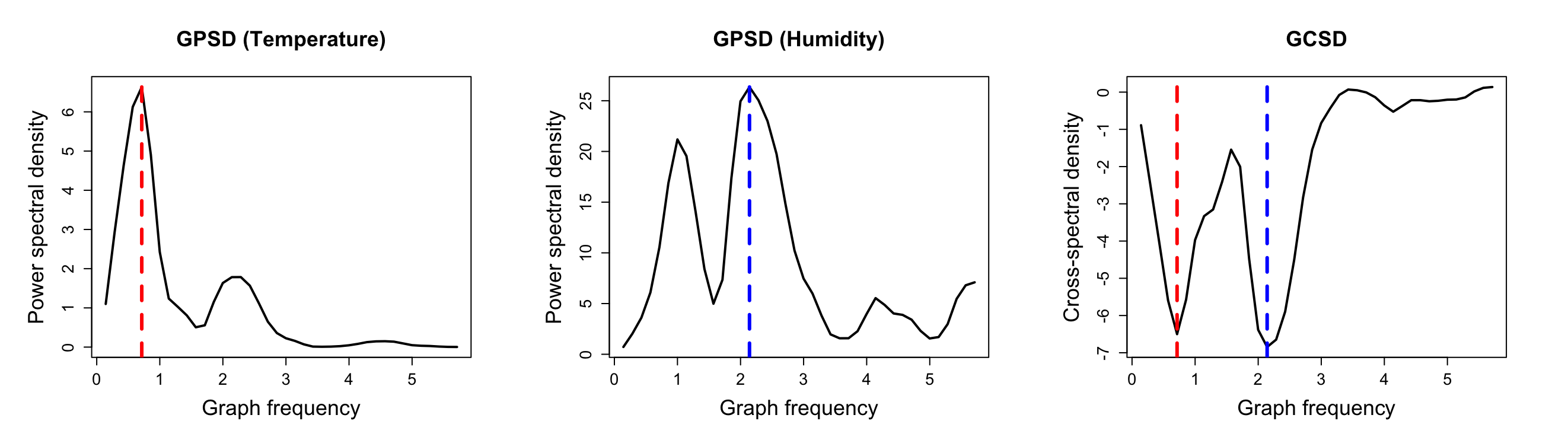}
	\caption{Cross-spectral analysis results on the Brest meteorological dataset: windowed graph Fourier transform-based estimates of GPSD for temperature (left), humidity (middle), and GCSD (right).}
	\label{fig:france_cpsd}
\end{figure}

\begin{figure}[!t]
	\centering
	\includegraphics[width=0.8\textwidth]{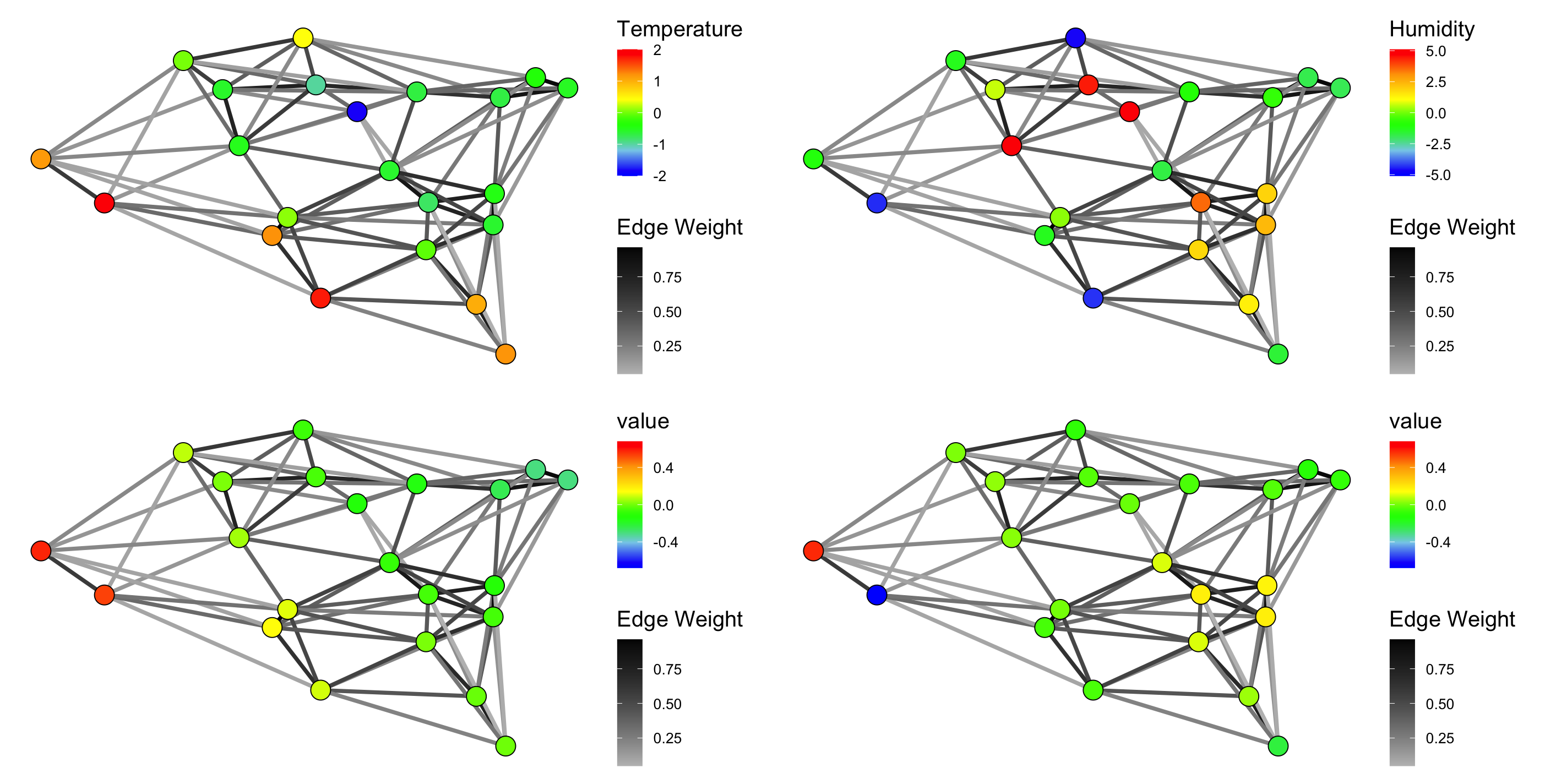}
	\caption{Visualization of signals on the Brest weather network: centered monthly average temperature for January  (top left), centered monthly average humidity for January (top right),  $\bs{v}_{3}^{\text{brest}}$ (bottom left), and $\bs{v}_{7}^{\text{brest}}$ (bottom right).}
	\label{fig:france_meteo}
\end{figure}

\Cref{fig:france_cpsd} shows the estimation results. The GPSD estimates for temperature and humidity peak at 0.714 and 2.142, aligning closely with $\lambda_3^{\text{brest}}$ and $\lambda_7^{\text{brest}}$. This suggests that the temperature exhibits a predominant pattern associated with $\bs{v}_3^{\text{brest}}$, while the humidity follows a predominant pattern linked to $\bs{v}_7^{\text{brest}}$. As shown in \Cref{fig:france_meteo}, this observation is plausible. Cross-spectral analysis provides additional insights into the relationship between two measurements. The negative peaks of 0.714 and 2.142 in the GCSD estimate imply that temperature and humidity share two primary components, $\bs{v}_{3}^{\text{brest}}$ and $\bs{v}_{7}^{\text{brest}}$, respectively, with opposite signs. This is consistent with a simple visual inspection plotting temperature and humidity across node index, which shows the opposite behavioral trends.

\section{Robust spectral analysis on graphs} \label{sec:robustspectral}
In this section, we discuss robust spectral analysis on graphs, proposing $M$-type estimators of GPSD and GCSD. It is natural to robustify the graph periodogram and the windowed average graph periodogram by utilizing the Huber loss function $\rho_c$ \cite{Huber1964} as 		
\begin{equation*}
	\rho_c(t)=\begin{cases}
		t^2 & \text{if } \lvert t\rvert \le c, \\ 
		2c(\lvert t \rvert-c/2) & \text{if } \lvert t\rvert > c,
	\end{cases}
\end{equation*}
for a cutoff constant $c>0$. By replacing $L_2$ loss in (\ref{eq:graph_ls}) with the Huber loss, an $M$-type estimator $\hat{\bs{p}}_{X,c}$ of GPSD can be defined as
\begin{equation*} \label{huberoptimizegraph_psd}
	\hat{\bs{p}}_{X,c} = \argmin_{\bs{p}} \sum\limits_{i=1}^{N^2}  \rho_c\left(( \hat{\bs{\sigma}}_{X} - \bs{G} \bs{p})_i \right).
\end{equation*}

In addition, the windowed average graph periodogram can be considered as a least squares solution. To be specific, for $\bs{x}_{\bs{w}_m}:= \operatorname{diag}(\bs{w}_m)\bs{x}$, $~~\hat{\bs{p}}^{\mathcal{W}}_{X}$ of (\ref{eq:window_avg_graph_periodogram}) is equal to $\frac{1}{M} \sum_{m=1}^M \lvert \bs{V}^H \bs{x}_{\bs{w}_m}\rvert^2$; thus, it is a least squares solution of $\bs{p}$ that minimizes $\lVert \hat{\bs{\sigma}}_{X}^{\mathcal{W}} - \bs{G} \bs{p} \rVert_2^2$, where $\hat{\bs{\sigma}}_{X}^{\mathcal{W}} = \operatorname{vec}\big( \hat{\bs{\Sigma}}_{X}^{\mathcal{W}}\big)$ and  $\hat{\bs{\Sigma}}_{X}^{\mathcal{W}} = \frac{1}{M}\sum_{m=1}^{M} \bs{x}_{\bs{w}_m} \bs{x}_{\bs{w}_m}^H$. Here, $\bs{x}_{\bs{w}_m}:= \operatorname{diag}(\bs{w}_m)\bs{x}$ and $\bs{y}_{\bs{w}_m}:= \operatorname{diag}(\bs{w}_m)\bs{y}$. Therefore, we define an $M$-type estimator of GPSD as 
\begin{equation*} \label{huberoptimizegraph_psd_window}
	\hat{\bs{p}}_{X,c}^{\mathcal{W}} = \argmin_{\bs{p}} \sum\limits_{i=1}^{N^2} \rho_c(( \hat{\bs{\sigma}}_{X}^{\mathcal{W}} - \bs{G} \bs{p})_i).
\end{equation*}

Similarly, we obtain a robust graph cross-periodogram as
\begin{equation*} \label{huberoptimizegraph_cpsd}
	\hat{\bs{p}}_{XY,c} = \argmin_{\bs{p}} \sum\limits_{i=1}^{N^2} \rho_c(( \hat{\bs{\sigma}}_{XY} - \bs{G} \bs{p})_i).
\end{equation*}
We also define an $M$-type estimator by modifying the windowed average graph cross-periodogram as 
\begin{equation*} \label{huberoptimizegraph_cpsd_window}
	\hat{\bs{p}}_{XY,c}^{\mathcal{W}} = \argmin_{\bs{p}} \sum\limits_{i=1}^{N^2} \rho_c(( \hat{\bs{\sigma}}_{XY}^{\mathcal{W}} - \bs{G} \bs{p})_i),
\end{equation*}
where $\hat{\bs{\sigma}}_{XY}^{\mathcal{W}} = \operatorname{vec}\big( \hat{\bs{\Sigma}}_{XY}^{\mathcal{W}} \big)$ and $\hat{\bs{\Sigma}}_{XY}^{\mathcal{W}} = \sum_{m=1}^{M} \bs{x}_{\bs{w}_m} \bs{y}_{\bs{w}_m}^H$. 

We demonstrate the robustness of the proposed $M$-type windowed average graph periodogram and $M$-type windowed average graph cross-periodogram by simulating signals with outliers on the Karate club network. We compare $M$-type windowed average estimators with the original windowed average estimators. We use the graph Laplacian as the graph shift operator and utilize 100 random windows. When calculating $M$-type windowed average estimators, we use the Huber loss function $\rho_c$ with $c=0.25$, and the optimization solution is computed using the iteratively reweighted least squares method.	

The left panel of \Cref{fig:cpsd_simul_window_robust} shows the robustness of the $M$-type windowed average graph periodogram compared to the original windowed average graph periodogram. The signal is generated as $\bs{x}_{\text{karate}} = 3\bs{v}_{20}^{\text{karate}}$, where $\bs{v}_{i}^{\text{karate}}$ denotes the $i$th eigenvector of the graph Laplacian corresponding to the $i$th eigenvalue $\lambda_{i}^{\text{karate}}$. A signal with an outlier is obtained by setting the signal value of the 25th node to 4, which is greater than the maximum signal value of the original signal. As expected, the windowed average graph periodogram peaks at $\lambda_{20}^{\text{karate}}$ when computed using the data without outliers, but the estimates become inaccurate when calculated using the data with outliers. On the other hand, the $M$-type windowed average estimate is robust to the outlier, closely resembling the results of the windowed average graph periodogram obtained from the data without outliers. 

We now evaluate the robustness of the $M$-type windowed average graph cross-periodogram. To this end, we generate two signals as $\bs{x}_{\text{karate}} = 5\bs{v}_{20}^{\text{karate}}+20\bs{v}_{30}^{\text{karate}}$ and $\bs{y}_{\text{karate}} = 5\bs{v}_{20}^{\text{karate}}+20\bs{v}_{10}^{\text{karate}}$. Signals with outliers are obtained by setting the signal value of $\bs{x}_{\text{karate}}$ at the 25th node to -10 and the signal value of $\bs{y}_{\text{karate}}$ at the 15th node to 10, respectively. The right panel of \Cref{fig:cpsd_simul_window_robust} shows the robustness of the $M$-type windowed average graph cross-periodogram. As one can see, the $M$-type windowed average estimate correctly captures the common graph frequency  $\lambda_{20}^{\text{karate}}$ between the two signals. %From these examples, we can conclude the proposed M-type windowed average estimators provide robust estimation in the presence of outliers within signals.
\begin{figure}
	\centering
	\includegraphics[width=0.95\textwidth]{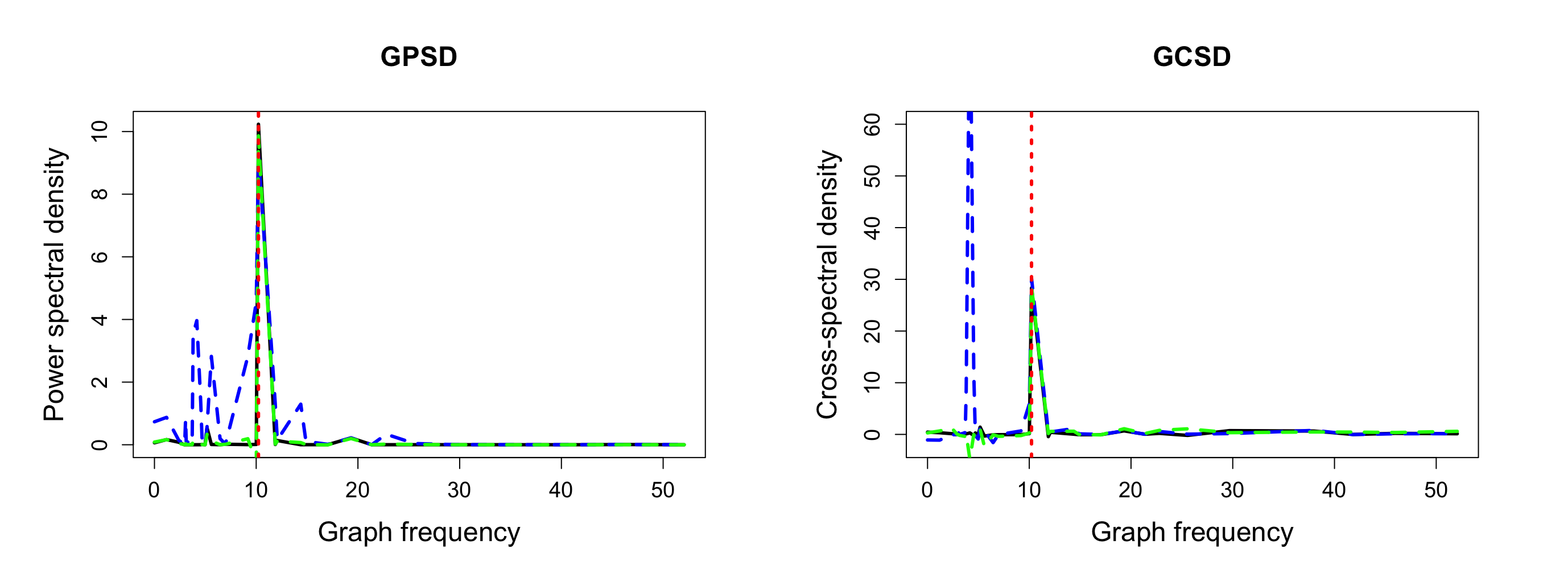}
	\caption{Estimation results with outliers: GPSD estimation (left) and GCSD estimation (right). The windowed average estimate using data without outliers (black), the estimate using data with outliers (blue dashed), the $M$-type windowed average estimate using data with outliers (green dashed), and frequency corresponding to $\lambda_{20}^{\text{karate}}$ (red dotted).}
	\label{fig:cpsd_simul_window_robust}
\end{figure}

We leave the theoretical properties of the $M$-type estimators to future research. 

\section{Concluding remarks}  \label{sec:conclusion}
In this study, we explored cross-spectral analysis tools for graph signals. We defined concepts such as joint weak stationarity, cross-spectral density, and coherence on graphs. We proposed several estimators for graph cross-spectral density and provided theoretical properties related to the definitions and estimators. Through numerical experiments on various irregular graphs, we investigated the usefulness and performance of these estimators. The results demonstrated the effectiveness of cross-spectral analysis in identifying shared behaviors between two graph processes that cannot be easily achieved solely through visualization or spectral analysis of a single graph process. We further suggested robustifying the proposed estimators to make them robust to outliers. This work contributes to the emerging field of multivariate graph signal processing, which has yet to be extensively studied.

There are several avenues for further development in this area. First, designing optimal windows for evaluating the windowed average graph cross-periodograms remains an open question. \textcolor{black}{Inspired by the parametric estimators of GPSDs proposed in \cite{Segarra2018} and \cite{Marques2017}, parametric estimators of GCSDs offer another promising direction. In addition, the study of non-stationary graph processes allows further research.} Locally stationary graph processes introduced by \cite{Canbolat2024} can be a promising starting point. Furthermore, developing a wavelet-based graph periodogram or graph cross-periodogram for such processes can be pursued in future work. 
\textcolor{black}{Finally, we can extend our framework to multivariate graph signals. In high-dimensional settings where each node is associated with $p$ variables and $p$ is sufficiently large, there is a strong likelihood that some variables can be represented as linear combinations of others. This can cause a singularity in the pairwise cross-spectral density matrix, potentially leading to computational and stability problems. To address this, one potential approach is to incorporate shrinkage methods for spectral density estimation, as developed in the context of time series and Euclidean spaces \cite{Fiecas2014data, Botts2006shrinkage, Abonazel2022developing}.}
These approaches can potentially advance our understanding of graph signal processing and are left as future work.	
	
	\section*{Acknowledgments}
This research was supported by the National Research Foundation of Korea (NRF) funded by the Korea government (2021R1A2C1091357).
	
	\clearpage
	\begin{appendices}
		\section{Proofs} \label{appendix:proof}
		\subsection{Proof of \Cref{prop:jointstationarity_defequiv}} \label{pf:jointstationarity_defequiv}
		\begin{proof}
			We follow similar processes of proofs outlined in \cite{Marques2017}. The equivalence between \Cref{def:jointweakstationary}(b) and (c) can be readily demonstrated by following the proof provided in \cite{Marques2017}. To prove the equivalence between \Cref{def:jointweakstationary}(a) and (c), let us first assume that \Cref{def:jointweakstationary}(a) holds. Given $\bs{S} = \bs{V} \bs{\Lambda} \bs{V}^H$, let us define $\tilde{\bs{h}}_i:=\operatorname{diag}(\sum_{\ell=0}^{N-1} h_{i,\ell} \bs{S}^\ell)$ and $\bs{h}_i:= (h_{i,0}, \ldots, h_{i,N-1})^\top$ ($i=1,2$). We further let $\lambda_j = (\bs{\Lambda})_{j,j}$ and $\bs{\Psi}$ be the $N \times N$ Vandermonde matrix with entries $\psi_{jk} = \lambda_j^{k-1}$. Then, we can rewrite $\bs{\mathrm{H}}_i$ ($i=1,2$) as
			\begin{equation*}
				\bs{\mathrm{H}}_i = \sum\limits_{\ell=0}^{N-1} h_{i,\ell} \bs{S}^\ell = \bs{V} \operatorname{diag}(\tilde{\bs{h}}_i) \bs{V}^H = \bs{V} \operatorname{diag}(\bs{\Psi}\bs{h}_i) \bs{V}^H.
			\end{equation*}
			Consequently, by letting $\bs{V}^H Cov(\bs{\epsilon}_1, \bs{\epsilon}_2) \bs{V} = \operatorname{diag}(\tilde{\bs{\epsilon}})$,
			% \begin{gather*}
				\begin{equation*}
					\begin{aligned}
						\bs{\Sigma}_{X} &= E\left((\bs{\mathrm{H}}_1\bs{\epsilon}_1) (\bs{\mathrm{H}}_1\bs{\epsilon}_1)^H\right) = \bs{\mathrm{H}}_1 E(\bs{\epsilon}_1 \bs{\epsilon}_1^H)\bs{\mathrm{H}}_1^H \\
						&= \bs{\mathrm{H}}_1 \bs{\mathrm{H}}_1^H = \bs{V} \operatorname{diag}(\lvert \tilde{\bs{h}}_1 \rvert^2) \bs{V}^H, \\
						\bs{\Sigma}_{Y} &= E\left((\bs{\mathrm{H}}_2\bs{\epsilon}_2) (\bs{\mathrm{H}}_2\bs{\epsilon}_2)^H\right) = \bs{\mathrm{H}}_2 E(\bs{\epsilon}_2 \bs{\epsilon}_2^H)\bs{\mathrm{H}}_2^H \\
						&= \bs{\mathrm{H}}_2 \bs{\mathrm{H}}_2^H = \bs{V} \operatorname{diag}(\lvert \tilde{\bs{h}}_2 \rvert^2) \bs{V}^H, \\
						\bs{\Sigma}_{XY} &= E\left((\bs{\mathrm{H}}_1\bs{\epsilon}_1) (\bs{\mathrm{H}}_2\bs{\epsilon}_2)^H\right) = \bs{\mathrm{H}}_1 E(\bs{\epsilon}_1 \bs{\epsilon}_2^H)\bs{\mathrm{H}}_2^H \\ &=  \bs{V} \operatorname{diag}(\tilde{\bs{\gamma}})\bs{V}^H,
					\end{aligned}
				\end{equation*}
				% \end{gather*}
			where $\tilde{\bs{\gamma}} = \tilde{\bs{h}}_1 \circ \tilde{\bs{\epsilon}} \circ \tilde{\bs{h}}_2^*$. This leads to \Cref{def:jointweakstationary}(c). Conversely, if \Cref{def:jointweakstationary}(c) holds, there exist $N \times 1$ vectors $\tilde{\bs{\gamma}}_{1}$, $\tilde{\bs{\gamma}}_{2}$, and $\tilde{\bs{\gamma}}_{12}$ satisfying $\bs{\Sigma}_{X} = \bs{V}\operatorname{diag}(\tilde{\bs{\gamma}}_{1}) \bs{V}^H$, $\bs{\Sigma}_{Y} = \bs{V}\operatorname{diag}(\tilde{\bs{\gamma}}_{2}) \bs{V}^H$, and $\bs{\Sigma}_{XY} = \bs{V}\operatorname{diag}(\tilde{\bs{\gamma}}_{12}) \bs{V}^H$. Here, $\tilde{\bs{\gamma}}_{1}$ and $\tilde{\bs{\gamma}}_{2}$ are elementwise nonnegative vectors because $\bs{\Sigma}_{X}$ and $\bs{\Sigma}_{Y}$ are positive semi-definite matrices. Since the eigenvalues of $\bs{S}$ are all distinct and $\bs{\Psi}$ is a full rank $N \times N$ Vandermonde matrix, there exist unique solutions $\bs{h}_1$ and $\bs{h}_2$ satisfying $\bs{\Psi} \bs{h}_1 = \sqrt{\tilde{\bs{\gamma}}_{1}}$ and $\bs{\Psi} \bs{h}_2 = \sqrt{\tilde{\bs{\gamma}}_{2}}$, where the square root is applied elementwise. Then, by selecting $\tilde{\bs{\epsilon}}$ such that  $\tilde{\bs{\epsilon}}= \frac{\tilde{\bs{\gamma}}_{12}}{\sqrt{\tilde{\bs{\gamma}}_{1}} \circ \sqrt{\tilde{\bs{\gamma}}_{2}}}$, and choosing $\bs{\epsilon}_1$ and $\bs{\epsilon}_2$  to satisfy $\bs{V}^H Cov(\bs{\epsilon}_1, \bs{\epsilon}_2) \bs{V} = \operatorname{diag}(\tilde{\bs{\epsilon}})$,  \Cref{def:jointweakstationary}(a) holds for the corresponding $\bs{\mathrm{H}}_1$, $\bs{\mathrm{H}}_2$, $\bs{\epsilon}_1$, and $\bs{\epsilon}_2$. 
		\end{proof}
		
		\subsection{Proof of \Cref{prop:spectralconvolution_cpsd}} \label{pf:spectralconvolution_cpsd}
		\begin{proof}
			We have $\bs{\mathrm{H}}_i = \bs{V} \operatorname{diag}(\tilde{\bs{h}}_i) \bs{V}^H$ ($i=1,2$) and $\bs{\Sigma}_{XY} = \bs{V}\operatorname{diag}(\bs{p}_{XY}) \bs{V}^H$. By using these three equations, we can derive 
			\begin{equation*}
				\bs{\Sigma}_{ZW} = \bs{\mathrm{H}}_1 \bs{\Sigma}_{XY} \bs{\mathrm{H}}_2^H = \bs{V} \operatorname{diag}(\tilde{\bs{h}}_1) \operatorname{diag}(\bs{p}_{XY}) \operatorname{diag}(\tilde{\bs{h}}_2^*) \bs{V}^H.
			\end{equation*} It completes the proofs for (a) and (b).
		\end{proof}
		
		\subsection{Proof of \Cref{prop:crosscov_element}} \label{pf:crosscov_element}
		\begin{proof}
			Since $\bs{\mathrm{H}}_1$ and $\bs{\mathrm{H}}_2$ are linear shift-invariant, by \cite{Sandryhaila2013}, we can rewrite $\bs{\mathrm{H}}_1$ and $\bs{\mathrm{H}}_2$ as
			\begin{equation*}
				\bs{\mathrm{H}}_i = \sum\limits_{\ell=0}^{L_i-1} h_{i,\ell} \bs{S}^\ell, \quad i=1,2.
			\end{equation*}
			Since $S_{i,j}$ is nonzero only if $i=j$ or the $i$th node and the $j$th node are connected, it is clear that $(\bs{S}^\ell)_{i,j}$ is nonzero only if the distance between the $i$th node and the $j$th node is equal to or less than $\ell$. Based on the fact that $\bs{\Sigma}_{XY} = \bs{\mathrm{H}}_1 \bs{\mathrm{H}}_2^H$, and that it is a polynomial of $\bs{S}$ and $\bs{S}^H$ of order $L_1 + L_2 -2$, the proof is completed.
		\end{proof}

		\subsection{Proof of \Cref{prop:gftuncorr_cpsd}} \label{pf:gftuncorr_cpsd}
		\begin{proof}
			We can compute the covariance matrix of $\tilde{X}$ as follows, and the proof is completed.
			\begin{equation*}
				\begin{aligned}
					\bs{\Sigma}_{\tilde{X}\tilde{Y}} &= E(\tilde{X}\tilde{Y}^H) = \bs{V}^H E(XY^H)\bs{V} \\ &= \bs{V}^H \bs{\Sigma}_{XY}\bs{V} = \operatorname{diag}(\bs{p}_{XY}).
				\end{aligned}
			\end{equation*}
		\end{proof}
		
		\subsection{Proof of \Cref{prop:graph_crossperiodogram_equiv}} \label{pf:graph_crossperiodogram_equiv}
		\begin{proof}
			Denote the $i$th column and the $(i,j)$th element of the matrix $\bs{V}$ by $\bs{v}_i$ and $v_{ij}$, respectively.
			We show that the $i$th element of $\hat{\bs{p}}_{XY}^{p}$ is equal to that of $\hat{\bs{p}}_{XY}^{c}$ as 
			\begin{equation*}
				\begin{aligned}
					(\hat{\bs{p}}_{XY}^{c})_i &= \left(\bs{V}^H\left(\frac{1}{R}\sum\limits_{r=1}^{R} \bs{x}_r \bs{y}_r^H\right) \bs{V}\right)_{i,i} = \frac{1}{R}\sum\limits_{r=1}^{R} \bs{v}_i^H \bs{x}_r \bs{y}_r^H \bs{v}_i \\ &= \frac{1}{R}\sum\limits_{r=1}^{R}  (\bs{v}_i^H \bs{x}_r) (\bs{v}_i^H \bs{y}_r)^* = (\hat{\bs{p}}_{XY}^{p})_i. 
				\end{aligned}
			\end{equation*}
			Therefore, $\hat{\bs{p}}_{XY}^{p} = \hat{\bs{p}}_{XY}^{c}$. Now, we show that the $i$th element of $\hat{\bs{p}}_{XY}^{c}$ is equal to that of $\hat{\bs{p}}_{XY}^{ls}$.
			It can be easily shown that $\bs{G}^H \bs{G} = I$. Thus, $\hat{\bs{p}}_{XY}^{ls} = \bs{G}^H \hat{\bs{\sigma}}_{XY}$. Then, 
			\begin{equation*}
				\begin{aligned}
					(\hat{\bs{p}}_{XY}^{ls})_i &= (\bs{G}^H \hat{\bs{\sigma}}_{XY})_i = \left(v_{i1}\bs{v}_i^H, \ldots,v_{iN}\bs{v}_i^H\right)\hat{\bs{\sigma}}_{XY} \\ &= \bs{v}_i^H \hat{\bs{\Sigma}}_{XY} \bs{v}_i = (\hat{\bs{p}}_{XY}^{c})_i.
				\end{aligned}
			\end{equation*}
			Hence, $\hat{\bs{p}}_{XY}^{c} = \hat{\bs{p}}_{XY}^{ls}$. The proof is completed. 
		\end{proof}

		\subsection{Proof of \Cref{prop:bias_var_graph_cross_periodogram}} \label{pf:bias_var_graph_cross_periodogram}
		To prove \Cref{prop:bias_var_graph_cross_periodogram}, we first present and prove several lemmas and theorems.
		\begin{lemma} \label{lem:isserlis}
			(Isserlis' Theorem, \cite{Isserlis1918}) \\
			If $(Z_1, \ldots, Z_n)^\top$ is a zero-mean real (or complex)-valued multivariate normal random vector, then the expectation of the product of random vector elements can be computed as
			\begin{equation*}
				E(Z_1Z_2\cdots Z_n) = \sum\limits_{P \in \mathbb{P}_n} \prod \limits_{\{i,j\} \in P} E(Z_i Z_j),
			\end{equation*}
			where $\mathbb{P}_n$ is the set of all distinct ways of partitioning $\{1,\ldots, n\}$ into pairs. Specifically, $\mathbb{P}_n = \{\{P_1, \ldots, P_{n/2}\} \, \vert \, P_1,\ldots, P_{n/2}$ are disjoint sets, each containing two elements, and  $P_1 \cup \cdots \cup P_{n/2} = \{1,\ldots, n\}\}$ if $n$ is even; otherwise, $\mathbb{P}_n$ is an empty set.
		\end{lemma}
		\begin{proof}
			Refer to \cite{Isserlis1918}.
		\end{proof}
		
		\begin{lemma} \label{lem:cor_isserlis}
			Let $\bs{u}$ and $\bs{v}$ be real-valued multivariate normal random vectors with $\bs{u} \sim N(\bs{0}, \bs{\Sigma_u})$, $\bs{v} \sim N(\bs{0}, \bs{\Sigma_v})$ and $Cov(\bs{u}, \bs{v}) = \bs{\Sigma_{uv}}$. For real (or complex) matrices $\bs{A}, \bs{B}, \bs{C}, \bs{D}$, and real (or complex) vectors $\bs{a}, \bs{b}, \bs{c}, \bs{d}$, the following equation holds.
			\begin{equation*}
				\begin{aligned}
					E&(\bs{A}\bs{u}\bs{v}^\top \bs{B}^\top \bs{C}\bs{u}\bs{v}^\top \bs{D}^\top) \\ &= \bs{A}\bs{\Sigma_{uv}}\bs{B}^\top \bs{C} \bs{\Sigma_{uv}} \bs{D}^\top + \bs{A}\bs{\Sigma_{u}}\bs{C}^\top \bs{B} \bs{\Sigma_{v}} \bs{D}^\top \\ & \quad + \bs{A}\bs{\Sigma_{uv}}\bs{D}^\top \operatorname{tr}(\bs{B}^\top \bs{C}\bs{\Sigma_{uv}}).
				\end{aligned}			
			\end{equation*}
		\end{lemma}
		\begin{proof}
			It is sufficient to show that the $(1,1)$th element of both sides of the equation are equal. Denote the $i$th rows of $\bs{A}, \bs{B}, \bs{C}, \bs{D}$ by $\bs{A}_i^\top, \bs{B}_i^\top, \bs{C}_i^\top, \bs{D}_i^\top$, respectively, and the $i$th element of $\bs{A}\bs{u}, \bs{B}\bs{v}, \bs{C}\bs{u}, \bs{D}\bs{v}$ by $a_i, b_i, c_i, d_i$, respectively. Additionally, let $\bs{w} = (\bs{u}^\top, \bs{v}^\top)^\top$. Then, the $(1,1)$th element of the left-hand side becomes $\sum_{i=1}^n a_1 b_i c_i d_1$. We note that 
			\begin{equation*}
				(a_1, b_i, c_i, d_1)^\top = \left(
				\begin{array}{cc}
					\bs{A}_1^\top & \bs{0}^\top  \\ [0.05cm] \hline \\ [-0.35cm]
					\bs{0}^\top & \bs{B}_i^\top \\ [0.05cm] \hline \\ [-0.35cm]
					\bs{C}_i^\top & \bs{0}^\top \\ [0.05cm] \hline \\ [-0.35cm]
					\bs{0}^\top & \bs{D}_1^\top
				\end{array}
				\right) \bs{w} \sim N(\bs{0}, \cdot).
			\end{equation*}
			Thus, by \Cref{lem:isserlis}, $E(a_1 b_i c_i d_1) = E(a_1 b_i)E(c_i d_1) + E(a_1 c_i)E(b_i d_1) + E(a_1 d_1)E(b_i c_i)$. The summation of the first term in the right-hand side over $i$, $\sum_{i=1}^{n} E(a_1 b_i)E(c_i d_1)$, becomes the $(1,1)$th element of $E(\bs{A}\bs{u}\bs{v}^\top \bs{B}^\top)E(\bs{C}\bs{u}\bs{v}^\top \bs{D}^\top)$, which is equal to the $(1,1)$th element of $\bs{A}\bs{\Sigma_{uv}}\bs{B}^\top \bs{C} \bs{\Sigma_{uv}} \bs{D}^\top$. Similarly, $\sum_{i=1}^{n} E(a_1 c_i)E(b_i d_1)$ is equal to the $(1,1)$th element of $\bs{A}\bs{\Sigma_{u}}\bs{C}^\top \bs{B} \bs{\Sigma_{v}} \bs{D}^\top$. Finally, $\sum_{i=1}^{n} E(a_1 d_1)E(b_i c_i) = E(a_1 d_1)\sum_{i=1}^{n}E(b_i c_i)$ is equal to the $(1,1)$th element of $E(\bs{A}\bs{u} \bs{v}^\top \bs{D}^\top) E(\bs{v}^\top \bs{B}^\top \bs{C} \bs{u})$, which is equal to the $(1,1)$th element of $\bs{A}\bs{\Sigma_{uv}}\bs{D}^\top \operatorname{tr}(\bs{B}^\top \bs{C}\bs{\Sigma_{uv}})$. Therefore, the given equation holds.
		\end{proof}
		
		\begin{theorem} \label{thm:fourth_moment}
			Let $\bs{x}$ and $\bs{y}$ be real-valued multivariate normal random vectors with $\bs{x} \sim N(\bs{\mu_{x}}, \bs{\Sigma_x})$, $\bs{y} \sim N(\bs{\mu_{y}}, \bs{\Sigma_y})$, and $Cov(\bs{x}, \bs{y}) = \bs{\Sigma_{xy}}$. For real (or complex) matrices $\bs{A}, \bs{B}, \bs{C}, \bs{D}$, and real (or complex) vectors $\bs{a}, \bs{b}, \bs{c}, \bs{d}$, the following equation holds.
			\begin{equation*}
				\begin{aligned}
					E&\left((\bs{A}\bs{x}+\bs{a})(\bs{B}\bs{y}+\bs{b})^\top (\bs{C}\bs{x}+\bs{c})(\bs{D}\bs{y}+\bs{d})^\top\right) \\
					& =  (\bs{A}\bs{\Sigma_{xy}}\bs{B}^\top + (\bs{A}\bs{\mu_x}+\bs{a})(\bs{B}\bs{\mu_y}+\bs{b})^\top) \\
					& \quad \quad \; \times (\bs{C} \bs{\Sigma_{xy}} \bs{D}^\top + (\bs{C}\bs{\mu_x} + \bs{c}) (\bs{D}\bs{\mu_y} + \bs{d})^\top)  \\
					& \quad \; + (\bs{A}\bs{\Sigma_{x}}\bs{C}^\top + (\bs{A}\bs{\mu_x} + \bs{a})(\bs{C}\bs{\mu_x} + \bs{c})^\top) \\
					& \quad \quad \; \times (\bs{B} \bs{\Sigma_{y}} \bs{D}^\top + (\bs{B}\bs{\mu_y} + \bs{b}) (\bs{D}\bs{\mu_y} + \bs{d})^\top) \\
					&\quad \; + \operatorname{tr}(\bs{B}^\top \bs{C}\bs{\Sigma_{xy}})(\bs{A}\bs{\Sigma_{xy}}\bs{D}^\top + (\bs{A}\bs{\mu_x} + \bs{a})(\bs{D}\bs{\mu_y} + \bs{d})^\top) \\
					&\quad \; + (\bs{B}\bs{\mu_y} + \bs{b})^\top (\bs{C}\bs{\mu_x} + \bs{c}) \\ 
					& \quad \quad \; \times (\bs{A}\bs{\Sigma_{xy}}\bs{D}^\top - (\bs{A}\bs{\mu_x} + \bs{a})(\bs{D}\bs{\mu_y} + \bs{d})^\top).
				\end{aligned}
			\end{equation*}
		\end{theorem}
		
		\begin{proof}
			We follow a similar process of proof as described in \cite{brookes2020}. Let $\bs{u} = \bs{x} - \bs{\mu_x}$ and $\bs{v} = \bs{y} - \bs{\mu_y}$. Then, $\bs{u} \sim N(\bs{0}, \bs{\Sigma_x})$, $\bs{v} \sim N(\bs{0}, \bs{\Sigma_y})$, and $Cov(\bs{u}, \bs{v}) = \bs{\Sigma_{xy}}$. Let us compute $E\left((\bs{A}\bs{u}+\bs{a})(\bs{B}\bs{v}+\bs{b})^\top (\bs{C}\bs{u}+\bs{c})(\bs{D}\bs{v}+\bs{d})^\top\right)$ first. When we multiply out the quartic expression and compute expectations of each term, we can remove odd-powered terms involving $\bs{u}$ and $\bs{v}$ by applying \Cref{lem:isserlis}. Consequently, we get
			\begin{equation*}
				\begin{aligned}
					E&\left((\bs{A}\bs{u}+\bs{a})(\bs{B}\bs{v}+\bs{b})^\top (\bs{C}\bs{u}+\bs{c})(\bs{D}\bs{v}+\bs{d})^\top\right) \\ & = \; E(\bs{A}\bs{u}\bs{v}^\top \bs{B}^\top \bs{C}\bs{u}\bs{v}^\top \bs{D}^\top) + E(\bs{A}\bs{u}\bs{v}^\top \bs{B}^\top \bs{c} \bs{d}^\top) \\
					&\quad \; + E(\bs{A}\bs{u}\bs{b}^\top \bs{C} \bs{u} \bs{d}^\top)  + E(\bs{A}\bs{u}\bs{b}^\top \bs{c}\bs{v}^\top \bs{D}^\top) \\ & \quad \; + E(\bs{a}\bs{v}^\top \bs{B}^\top \bs{C}\bs{u}\bs{d}^\top) + E(\bs{a}\bs{v}^\top \bs{B}^\top \bs{c}\bs{v}^\top \bs{D}^\top) 
					\\
					&\quad \;  + E(\bs{a}\bs{b}^\top \bs{C}\bs{u}\bs{v}^\top \bs{D}^\top) + \bs{a}\bs{b}^\top \bs{c}\bs{d}^\top.
				\end{aligned}
			\end{equation*}
			By \Cref{lem:cor_isserlis}, the first term in the right-hand side of the above equation is equal to $\bs{A}\bs{\Sigma_{xy}}\bs{B}^\top \bs{C} \bs{\Sigma_{xy}} \bs{D}^\top + \bs{A}\bs{\Sigma_{x}}\bs{C}^\top \bs{B} \bs{\Sigma_{y}} \bs{D}^\top + \bs{A}\bs{\Sigma_{xy}}\bs{D}^\top \operatorname{tr}(\bs{B}^\top \bs{C}\bs{\Sigma_{xy}})$. By appropriately moving and transposing scalar factors and utilizing the property of the trace,  the other terms on the right-hand side can be easily computed. As a result, we obtain
			\begin{equation*}
				\begin{aligned}
					E&\left((\bs{A}\bs{u}+\bs{a})(\bs{B}\bs{v}+\bs{b})^\top (\bs{C}\bs{u}+\bs{c})(\bs{D}\bs{v}+\bs{d})^\top\right)  \\ &=  \bs{A}\bs{\Sigma_{xy}}\bs{B}^\top \bs{C} \bs{\Sigma_{xy}} \bs{D}^\top + \bs{A}\bs{\Sigma_{x}}\bs{C}^\top \bs{B} \bs{\Sigma_{y}} \bs{D}^\top \\ &\quad + \bs{A}\bs{\Sigma_{xy}}\bs{D}^\top \operatorname{tr}(\bs{B}^\top \bs{C}\bs{\Sigma_{xy}})  + \bs{A}\bs{\Sigma_{xy}} \bs{B}^\top \bs{c} \bs{d}^\top \\ & \quad + \bs{A}\bs{\Sigma_x} \bs{C}^\top \bs{b} \bs{d}^\top 
					+ \bs{b}^\top \bs{c}\bs{A}\bs{\Sigma_{xy}} \bs{D}^\top
					\\ & \quad  + \bs{a}\bs{d}^\top \operatorname{tr}(\bs{B}^\top \bs{C}\bs{\Sigma_{xy}}) + \bs{a}\bs{c}^\top \bs{B} \bs{\Sigma_y} \bs{D}^\top 
					+ \bs{a}\bs{b}^\top \bs{C}\bs{\Sigma_{xy}} \bs{D}^\top \\ & \quad + \bs{a}\bs{b}^\top \bs{c}\bs{d}^\top \\ 
					& =  (\bs{A}\bs{\Sigma_{xy}}\bs{B}^\top + \bs{a}\bs{b}^\top)(\bs{C} \bs{\Sigma_{xy}} \bs{D}^\top + \bs{c} \bs{d}^\top) \\ & \quad + (\bs{A}\bs{\Sigma_{x}}\bs{C}^\top  + \bs{a}\bs{c}^\top)(\bs{B} \bs{\Sigma_{y}} \bs{D}^\top + \bs{b} \bs{d}^\top) \\
					& \quad + \operatorname{tr}(\bs{B}^\top \bs{C}\bs{\Sigma_{xy}})(\bs{A}\bs{\Sigma_{xy}}\bs{D}^\top + \bs{a}\bs{d}^\top) \\ & \quad + \bs{b}^\top \bs{c} (\bs{A}\bs{\Sigma_{xy}}\bs{D}^\top - \bs{a}\bs{d}^\top).
				\end{aligned}
			\end{equation*}
			Now, to prove the original statement, we replace $\bs{a}, \bs{b}, \bs{c}, \bs{d}$ by $\bs{A} \bs{\mu_x} + \bs{a}, \bs{B} \bs{\mu_y} + \bs{b}, \bs{C} \bs{\mu_x} + \bs{c}, \bs{D} \bs{\mu_y} + \bs{d}$, respectively. By utilizing the fact that $\bs{A} \bs{x} + \bs{a} = \bs{A} \bs{u} + (\bs{A} \bs{\mu_x} + \bs{a})$, we complete the proof.
		\end{proof}
		
		\begin{corollary} \label{cor:fourth_moment}
			Let $\bs{x}$ and $\bs{y}$ be real-valued multivariate normal random vectors with $\bs{x} \sim N(\bs{0}, \bs{\Sigma_x})$, $\bs{y} \sim N(\bs{0}, \bs{\Sigma_y})$, and $Cov(\bs{x}, \bs{y}) = \bs{\Sigma_{xy}}$. For real (or complex) matrices $\bs{A}$, the following equation holds.
			\begin{equation*}
				\begin{aligned}
					E(\bs{x}\bs{y}^\top \bs{A} \bs{x}\bs{y}^\top) = & \bs{\Sigma_{xy}}\bs{A} \bs{\Sigma_{xy}}  + \bs{\Sigma_{x}}\bs{A}^\top \bs{\Sigma_{y}} 
					+ \operatorname{tr}(\bs{A}\bs{\Sigma_{xy}})\bs{\Sigma_{xy}}.
				\end{aligned}
			\end{equation*}
		\end{corollary}
		\begin{proof}
			It is clear from \Cref{thm:fourth_moment}.
		\end{proof}
		
		We can now prove \Cref{prop:bias_var_graph_cross_periodogram}.
		\begin{proof}
			(a) It is easy to show by the following, 
			\begin{equation*}
				E(\hat{\bs{p}}_{XY}) = \operatorname{diag}(\bs{V}^H E(\hat{\bs{\Sigma}}_{XY}) \bs{V}) = \operatorname{diag}(\bs{V}^H \bs{\Sigma}_{XY} \bs{V}) = \bs{p}_{XY}. 
			\end{equation*}
			(b) We start by computing 
			\begin{equation*}
				\begin{aligned}
					E&\left((\hat{\bs{p}}_{XY})(\hat{\bs{p}}_{XY})^H\right) \\ &= \frac{1}{R^2} \sum\limits_{r=1}^{R} \sum\limits_{r'=1}^{R} E\left(\operatorname{diag}(\bs{V}^H \bs{x}_r\bs{y}_r^H \bs{V}) \operatorname{diag}(\bs{V}^H \bs{x}_{r'}\bs{y}_{r'}^H \bs{V})^H\right).
				\end{aligned}
			\end{equation*}
			If $r \neq r'$, we obtain that 
			\begin{equation*}
				\begin{aligned}
					E&\left(\operatorname{diag}(\bs{V}^H \bs{x}_r\bs{y}_r^H \bs{V}) \operatorname{diag}(\bs{V}^H \bs{x}_{r'}\bs{y}_{r'}^H \bs{V})^H\right) \\
					&= E\left(\operatorname{diag}(\bs{V}^H \bs{x}_r\bs{y}_r^H \bs{V})\right) E\left(\operatorname{diag}(\bs{V}^H \bs{x}_{r'}\bs{y}_{r'}^H \bs{V})^H\right) \\ 
					& = \operatorname{diag}(\bs{V}^H \bs{\Sigma}_{XY} \bs{V})\operatorname{diag}(\bs{V}^H \bs{\Sigma}_{XY} \bs{V})^H \\
					& = \bs{p}_{XY} \bs{p}_{XY}^H.
				\end{aligned}
			\end{equation*}
			If $r = r'$, $E(\operatorname{diag}(\bs{V}^H \bs{x}_r\bs{y}_r^H \bs{V}) \operatorname{diag}(\bs{V}^H \bs{x}_{r'}\bs{y}_{r'}^H \bs{V})^H)$ has its $(i,j)$th element  equal to 
			$$E((\bs{v}_i^H \bs{x}_r\bs{y}_r^H \bs{v}_i)(\bs{v}_j^H \bs{x}_r\bs{y}_r^H \bs{v}_j)^*)=\bs{v}_i^H E(\bs{x}_r\bs{y}_r^H \bs{v}_i \bs{v}_j^\top \bs{x}_r^*\bs{y}_r^\top)\bs{v}_j^*,$$
			where $\bs{v}_i$ is the $i$th column of $\bs{V}$. Since $\bs{x}_r$ and $\bs{y}_r$ are real-valued, by \Cref{cor:fourth_moment},      
			\begin{equation*}
				\begin{aligned}
					E&(\bs{x}_r\bs{y}_r^H \bs{v}_i \bs{v}_j^\top \bs{x}_r^*\bs{y}_r^\top) = E(\bs{x}_r\bs{y}_r^\top \bs{v}_i \bs{v}_j^\top \bs{x}_r \bs{y}_r^\top) \\ & = \bs{\Sigma}_{XY} \bs{v}_i \bs{v}_j^\top \bs{\Sigma}_{XY} + \bs{\Sigma}_{X}  \bs{v}_j \bs{v}_i^\top \bs{\Sigma}_{Y} + \operatorname{tr}(\bs{v}_i \bs{v}_j^\top \bs{\Sigma}_{XY})\bs{\Sigma}_{XY}.
				\end{aligned}			
			\end{equation*}
			Since $\bs{S}$ is symmetric, $\bs{V}$ is real matrix; thus,  
			\begin{equation*}
				\begin{aligned}
					E&\left((\bs{v}_i^H \bs{x}_r\bs{y}_r^H \bs{v}_i)(\bs{v}_j^H \bs{x}_r\bs{y}_r^H \bs{v}_j)^* \right) \\ & = 
					\begin{cases} 
						2 \lvert(\bs{p}_{XY})_i \rvert^2 + (\bs{p}_X)_i(\bs{p}_Y)_i & \text{if } i = j, \\ 
						(\bs{p}_{XY})_i (\bs{p}_{XY}^*)_j & \text{if } i \neq j.
					\end{cases}
				\end{aligned}
			\end{equation*}
			Combining these results together, we can compute $E\left((\hat{\bs{p}}_{XY})(\hat{\bs{p}}_{XY})^H\right)$, and  obtain
			\begin{equation*}
				\begin{aligned}
					Var(\hat{\bs{p}}_{XY}) &= E\left((\hat{\bs{p}}_{XY})(\hat{\bs{p}}_{XY})^H\right) - \bs{p}_{XY} \bs{p}_{XY}^H \\ &= \frac{1}{R} \left(\operatorname{diag}(\lvert \bs{p}_{XY} \rvert^2) + \operatorname{diag}(\bs{p}_X)  \operatorname{diag}(\bs{p}_Y)\right).
				\end{aligned}
			\end{equation*}
		\end{proof}
		
		\subsection{Proof of \Cref{prop:windowcrossperiodogram_expectation}} \label{pf:windowcrossperiodogram_expectation}
		\begin{proof}
			We can rewrite $\hat{\bs{p}}_{XY}^{\bs{w}} = \operatorname{diag}(\tilde{\bs{W}} \bs{V}^H \bs{x}  \bs{y}^H \bs{V} \tilde{\bs{W}}^H)$. Then, the following equation can be obtained by taking the expectation
			\begin{equation*}
				\begin{aligned}
					E&(\hat{\bs{p}}_{XY}^{\bs{w}}) = \operatorname{diag}(\tilde{\bs{W}} \bs{V}^H \bs{\Sigma}_{XY} \bs{V} \tilde{\bs{W}}^H) \\ &= \operatorname{diag}(\tilde{\bs{W}} \operatorname{diag}(\bs{p}_{XY}) \tilde{\bs{W}}^H) = (\tilde{\bs{W}} \circ \tilde{\bs{W}}^*)\bs{p}_{XY}.
				\end{aligned}
			\end{equation*}
			The proof is completed.
		\end{proof}
		
		\subsection{Proof of \Cref{prop:bias_var_window_graph_cross}} \label{pf:bias_var_window_graph_cross}
		\begin{proof}
			(a) The proof for the expectation follows the same approach to the proof of \Cref{prop:windowcrossperiodogram_expectation}. \\
			(b) To compute the trace of the covariance matrix, we focus on the diagonal elements of $Var(\hat{\bs{p}}_{XY}^{\mathcal{W}})$. For $1 \le i \le N$, compute $E\Big(\Big((\hat{\bs{p}}_{XY}^{\mathcal{W}})(\hat{\bs{p}}_{XY}^{\mathcal{W}})^H\Big)_{i,i}\Big) = E\left((\hat{\bs{p}}_{XY}^{\mathcal{W}})_i(\hat{\bs{p}}_{XY}^{\mathcal{W}})^*_i\right)$ as
			\begin{equation*}
				\begin{aligned}
					E&\left((\hat{\bs{p}}_{XY}^{\mathcal{W}})_i(\hat{\bs{p}}_{XY}^{\mathcal{W}})^*_i\right) \\ &= E \bigg( 
					\frac{1}{M}\sum\limits_{m=1}^{M} (\tilde{\bs{W}}_{m,i}^\top \tilde{\bs{x}}) (\tilde{\bs{W}}_{m,i}^\top \tilde{\bs{y}})^* \\
					& \quad \quad \quad \times \frac{1}{M}\sum\limits_{m'=1}^{M} (\tilde{\bs{W}}_{m',i}^\top \tilde{\bs{x}})^* (\tilde{\bs{W}}_{m',i}^\top \tilde{\bs{y}}) \bigg) \\
					& = \frac{1}{M^2} \sum\limits_{m=1}^{M} \sum\limits_{m'=1}^{M} \tilde{\bs{W}}_{m,i}^\top E(\tilde{\bs{x}} \tilde{\bs{y}}^H \tilde{\bs{W}}_{m,i}^* \tilde{\bs{W}}_{m',i}^H \tilde{\bs{x}}^* \tilde{\bs{y}}^\top) \tilde{\bs{W}}_{m',i},
				\end{aligned}
			\end{equation*}
			where $\tilde{\bs{W}}_{m,i}^\top$ is the $i$th row vector of $\tilde{\bs{W}}_{m}$, $\tilde{\bs{x}} = \bs{V}^H \bs{x}$, and $\tilde{\bs{y}} = \bs{V}^H \bs{y}$. Since $\bs{S}$ is symmetric, $\bs{V}$ is real matrix. Therefore, $\tilde{\bs{x}}$ and $\tilde{\bs{y}}$ are zero-mean real-valued multivariate normal random vectors, allowing us to utilize \Cref{cor:fourth_moment} to evaluate $E(\tilde{\bs{x}} \tilde{\bs{y}}^H \tilde{\bs{W}}_{m,i}^* \tilde{\bs{W}}_{m',i}^H \tilde{\bs{x}}^* \tilde{\bs{y}}^\top)$, in conjunction with the results from Property 3 in \cite{Segarra2016} and \Cref{prop:gftuncorr_cpsd}.  
			\begin{equation*}
				\begin{aligned}
					E&(\tilde{\bs{x}} \tilde{\bs{y}}^H \tilde{\bs{W}}_{m,i}^* \tilde{\bs{W}}_{m',i}^H \tilde{\bs{x}}^* \tilde{\bs{y}}^\top) \\ & = \operatorname{diag}(\bs{p}_{XY}) \tilde{\bs{W}}_{m,i}^* \tilde{\bs{W}}_{m',i}^H \operatorname{diag}(\bs{p}_{XY}) \\ & \quad + \operatorname{diag}(\bs{p}_{X}) \tilde{\bs{W}}_{m',i}^* \tilde{\bs{W}}_{m,i}^H \operatorname{diag}(\bs{p}_{Y}) \\
					&\quad + \operatorname{tr}(\tilde{\bs{W}}_{m,i}^* \tilde{\bs{W}}_{m',i}^H \operatorname{diag}(\bs{p}_{XY})) \operatorname{diag}(\bs{p}_{XY}).
				\end{aligned}
			\end{equation*}
			Then, by appropriately transposing scalar factors and utilizing the property of the trace, we obtain
			\begin{equation*}
				\begin{aligned}
					&\tilde{\bs{W}}_{m,i}^\top E(\tilde{\bs{x}} \tilde{\bs{y}}^H \tilde{\bs{W}}_{m,i}^* \tilde{\bs{W}}_{m',i}^H \tilde{\bs{x}}^* \tilde{\bs{y}}^\top) \tilde{\bs{W}}_{m',i} \\ & \quad = \tilde{\bs{W}}_{m,i}^\top \operatorname{diag}(\bs{p}_{XY}) \tilde{\bs{W}}_{m,i}^* \tilde{\bs{W}}_{m',i}^H \operatorname{diag}(\bs{p}_{XY})\tilde{\bs{W}}_{m',i} \\
					&\quad \quad + \tilde{\bs{W}}_{m,i}^\top \operatorname{diag}(\bs{p}_{X}) \tilde{\bs{W}}_{m',i}^* \tilde{\bs{W}}_{m,i}^H \operatorname{diag}(\bs{p}_{Y}) \tilde{\bs{W}}_{m',i}\\
					&\quad \quad + \lvert \tilde{\bs{W}}_{m,i}^\top \operatorname{diag}(\bs{p}_{XY}) \tilde{\bs{W}}_{m',i}^*\rvert^2.
				\end{aligned}
			\end{equation*}
			Now, let us evaluate $\big\lvert E((\hat{\bs{p}}_{XY}^{\mathcal{W}})_i) \big\rvert^2$. 
			\begin{equation*}
				\begin{aligned}
					&\big\lvert E((\hat{\bs{p}}_{XY}^{\mathcal{W}})_i) \big\rvert^2 
					\\ & \quad = E \left( 
					\frac{1}{M}\sum\limits_{m=1}^{M} (\tilde{\bs{W}}_{m,i}^\top \tilde{\bs{x}}) (\tilde{\bs{W}}_{m,i}^\top \tilde{\bs{y}})^* \right) \\ & \quad \quad \times  E \left(\frac{1}{M}\sum\limits_{m'=1}^{M} (\tilde{\bs{W}}_{m',i}^\top \tilde{\bs{x}})^* (\tilde{\bs{W}}_{m',i}^\top \tilde{\bs{y}}) \right) \\
					& \quad = \frac{1}{M^2} \sum\limits_{m=1}^{M} \sum\limits_{m'=1}^{M} \tilde{\bs{W}}_{m,i}^\top \operatorname{diag}(\bs{p}_{XY}) \tilde{\bs{W}}_{m,i}^* \tilde{\bs{W}}_{m',i}^H \operatorname{diag}(\bs{p}_{XY}) \tilde{\bs{W}}_{m',i}.
				\end{aligned}
			\end{equation*}
			Thus,
			% \begin{equation*}
				\begin{alignat*}{2}
					&\left(Var(\hat{\bs{p}}_{XY}^{\mathcal{W}})\right)_{i,i}  = &&E\left((\hat{\bs{p}}_{XY}^{\mathcal{W}})_i(\hat{\bs{p}}_{XY}^{\mathcal{W}})^*_i\right) - \big\lvert E((\hat{\bs{p}}_{XY}^{\mathcal{W}})_i) \big\rvert^2  \\
					& \quad = \frac{1}{M^2} \sum\limits_{m=1}^{M} \sum\limits_{m'=1}^{M}&&\big(\tilde{\bs{W}}_{m,i}^\top \operatorname{diag}(\bs{p}_{X}) \tilde{\bs{W}}_{m',i}^* \tilde{\bs{W}}_{m,i}^H \operatorname{diag}(\bs{p}_{Y}) \tilde{\bs{W}}_{m',i} \\ & &&+ \lvert \tilde{\bs{W}}_{m,i}^\top \operatorname{diag}(\bs{p}_{XY}) \tilde{\bs{W}}_{m',i}^*\rvert^2 \big) \\
					& \quad = \frac{1}{M^2} \sum\limits_{m=1}^{M} \sum\limits_{m'=1}^{M} &&\big(\langle (\tilde{\bs{W}}_{m,m'}\bs{p}_{X})_i, (\tilde{\bs{W}}_{m,m'}\bs{p}_{Y})_i \rangle \\ & && + \lvert (\tilde{\bs{W}}_{m,m'}\bs{p}_{XY})_i\rvert^2 \big).
				\end{alignat*}
				% \end{equation*}
			Therefore, the proof is completed as follows.
			% \begin{equation*}
				\begin{alignat*}{2}
					\operatorname{tr}&(Var(\hat{\bs{p}}_{XY}^{\mathcal{W}})) 
					\\ &= \frac{1}{M^2} \sum\limits_{m=1}^{M} \sum\limits_{m'=1}^{M} \operatorname{tr}&&\big((\tilde{\bs{W}}_{m,m'}\bs{p}_{XY})(\tilde{\bs{W}}_{m,m'}\bs{p}_{XY})^H \\ & &&+ (\tilde{\bs{W}}_{m,m'}\bs{p}_{X})(\tilde{\bs{W}}_{m,m'}\bs{p}_{Y})^H\big).
				\end{alignat*}
				% \end{equation*}
		\end{proof}
		
		\subsection{Proof of \Cref{prop:gcpsd_windowdesign}} \label{pf:gcpsd_windowdesign}
		\begin{proof}
			It is directly extended from the proof of Proposition 5 in \cite{Marques2017} by \Cref{prop:crosscov_element}.
		\end{proof}
		
		\subsection{Proof of \Cref{prop:coherence_ineq}} \label{pf:coherence_ineq}
		\begin{proof}
			For a given $1 \le i \le N$, we know that $\lvert (\bs{p}_{XY})_i \rvert^2 = \lvert (E(\tilde{X} \tilde{Y}^H))_{i,i} \rvert^2 = \lvert E(\tilde{X}_i \tilde{Y}_i^*) \rvert^2$ from \Cref{prop:gftuncorr_cpsd}, where $\tilde{X}_i$ and $\tilde{Y}_i$ represent the $i$th elements of the GFT processes $\tilde{X}$ and $\tilde{Y}$, respectively. Similarly, $(\bs{p}_{X})_i = E(\tilde{X}_i \tilde{X}_i^*)$ and $(\bs{p}_{Y})_i = E(\tilde{Y}_i \tilde{Y}_i^*)$. Therefore, it is clear that $0 \le \bs{c}_{XY}$. Moreover, by the Cauchy-Schwarz inequality, it also holds that $\bs{c}_{XY} \le 1$. 
			
			To satisfy the equality $\bs{c}_{XY}=\bs{0}$, $E(\tilde{X}_i \tilde{Y}_i^*)$ should be zero for all $i$. It is equivalent to $\bs{p}_{XY}=\bs{0}$ from \Cref{prop:gftuncorr_cpsd}, which is also equivalent to $\bs{\Sigma}_{XY}$ being a zero matrix. Thus, the equality  $\bs{c}_{XY}=\bs{0}$ holds when $X$ and $Y$ are uncorrelated. To satisfy the  equality $\bs{c}_{XY}=\bs{1}$, $\tilde{Y}_i = \alpha_i \tilde{X}_i$ for some constant $\alpha_i$ for all $i$. It is satisfied when $Y = \bs{\mathrm{H}} X$ for a filter $\bs{\mathrm{H}}$ with frequency response $\tilde{\bs{h}}_i = \alpha_i$, and the proof is complete.
		\end{proof}

	\end{appendices}
	\clearpage

	\bibliographystyle{apalike}
	\bibliography{refs}

@incollection{Segarra2018,
  title={Statistical graph signal processing: Stationarity and spectral estimation},
  author={Segarra, Santiago and Chepuri, Sundeep Prabhakar and Marques, Antonio G and Leus, Geert},
  booktitle={Cooperative and Graph Signal Processing},
  pages={325--347},
  year={2018},
  publisher={Elsevier}
}

@article{Marques2017,
  title={Stationary graph processes and spectral estimation},
  author={Marques, Antonio G and Segarra, Santiago and Leus, Geert and Ribeiro, Alejandro},
  journal={IEEE Transactions on Signal Processing},
  volume={65},
  number={22},
  pages={5911--5926},
  year={2017},
  publisher={IEEE}
}

@article{Kurokawa2017,
  title={Multi-dimensional graph \text{Fourier} transform},
  author={Kurokawa, Takashi and Oki, Taihei and Nagao, Hiromichi},
  journal={arXiv preprint arXiv:1712.07811},
  year={2017}
}

@article{Yan2022,
  title={Multi-dimensional graph fractional \text{Fourier} transform and its application to data compression},
  author={Yan, Fang-Jia and Li, Bing-Zhao},
  journal={Digital Signal Processing},
  volume={129},
  pages={103683},
  year={2022},
  publisher={Elsevier}
}

@inproceedings{Girault2015,
  title={Stationary graph signals using an isometric graph translation},
  author={Girault, Benjamin},
  booktitle={2015 23rd European Signal Processing Conference (EUSIPCO)},
  pages={1516--1520},
  year={2015},
  organization={IEEE}
}

@phdthesis{Girault2015thesis,
  title={Signal Processing on Graphs-Contributions to An Emerging Field},
  author={Girault, Benjamin},
  year={2015},
  school={Ecole normale sup{\'e}rieure de lyon-ENS LYON}
}

@article{Perraudin2017,
  title={Stationary signal processing on graphs},
  author={Perraudin, Nathana{\"e}l and Vandergheynst, Pierre},
  journal={IEEE Transactions on Signal Processing},
  volume={65},
  number={13},
  pages={3462--3477},
  year={2017},
  publisher={IEEE}
}

@article{Bartlett1950,
  title={Periodogram analysis and continuous spectra},
  author={Bartlett, Maurice S},
  journal={Biometrika},
  volume={37},
  number={1/2},
  pages={1--16},
  year={1950},
  publisher={JSTOR}
}

@article{Welch1967,
  title={The use of fast Fourier transform for the estimation of power spectra: A method based on time averaging over short, modified periodograms},
  author={Welch, Peter},
  journal={IEEE Transactions on Audio and Electroacoustics},
  volume={15},
  number={2},
  pages={70--73},
  year={1967},
  publisher={IEEE}
}

@article{Isserlis1918,
  title={On a formula for the product-moment coefficient of any order of a normal frequency distribution in any number of variables},
  author={Isserlis, Leon},
  journal={Biometrika},
  volume={12},
  number={1/2},
  pages={134--139},
  year={1918},
  publisher={JSTOR}
}

@misc{Brookes2020,
  Author = {Brookes, Mike},
  Title  = {The Matrix Reference Manual},
  Note   = {[online] \url{http://www.ee.imperial.ac.uk/hp/staff/dmb/matrix/intro.html}},
  year = 2020,
}

@article{Huber1964,
  title={Robust Estimation of a Location Parameter},
  author={Peter J. Huber},
  journal={Annals of Mathematical Statistics},
  volume={35},
  number={1},
  pages={73--101},
  year={1964}
}

@article{Sandryhaila2013,
  title={Discrete signal processing on graphs},
  author={Sandryhaila, Aliaksei and Moura, Jos{\'e} MF},
  journal={IEEE Transactions on Signal Processing},
  volume={61},
  number={7},
  pages={1644--1656},
  year={2013},
  publisher={IEEE}
}

@inproceedings{Sandryhaila2013GFT,
  title={Discrete signal processing on graphs: Graph \text{Fourier} transform},
  author={Sandryhaila, Aliaksei and Moura, Jos{\'e} MF},
  booktitle={2013 IEEE International Conference on Acoustics, Speech, and Signal Processing},
  pages={6167--6170},
  year={2013},
  organization={IEEE}
}

@article{Sandryhaila2014,
  title={Discrete signal processing on graphs: Frequency analysis},
  author={Sandryhaila, Aliaksei and Moura, Jose MF},
  journal={IEEE Transactions on Signal Processing},
  volume={62},
  number={12},
  pages={3042--3054},
  year={2014},
  publisher={IEEE}
}

@book{Blackman1959,
  title={The Measurement of Power Spectra From the Point of View of Communications Engineering},
  author={Blackman, R.B. and Tukey, J.W.},
  isbn={9780486605074},
  lccn={59010185},
  year={1959},
  publisher={Dover Publications}
}

@article{Shuman2016,
  title={Vertex-frequency analysis on graphs},
  author={Shuman, David I and Ricaud, Benjamin and Vandergheynst, Pierre},
  journal={Applied and Computational Harmonic Analysis},
  volume={40},
  number={2},
  pages={260--291},
  year={2016},
  publisher={Elsevier}
}

@inproceedings{Segarra2016,
  title={Stationary graph processes: Nonparametric spectral estimation},
  author={Segarra, Santiago and Marques, Antonio G and Leus, Geert and Ribeiro, Alejandro},
  booktitle={2016 IEEE Sensor Array and Multichannel Signal Processing Workshop (SAM)},
  pages={1--5},
  year={2016},
  organization={IEEE}
}

@article{Chen2015a,
  title={Signal recovery on graphs: Variation minimization},
  author={Chen, Siheng and Sandryhaila, Aliaksei and Moura, Jos{\'e} MF and Kova{\v{c}}evi{\'c}, Jelena},
  journal={IEEE Transactions on Signal Processing},
  volume={63},
  number={17},
  pages={4609--4624},
  year={2015},
  publisher={IEEE}
}

@article{Chen2015b,
  title={Discrete signal processing on graphs: Sampling theory},
  author={Chen, Siheng and Varma, Rohan and Sandryhaila, Aliaksei and Kova{\v{c}}evi{\'c}, Jelena},
  journal={IEEE Transactions on Signal Processing},
  volume={63},
  number={24},
  pages={6510--6523},
  year={2015},
  publisher={IEEE}
}

@article{Tanaka2018,
  title={Spectral domain sampling of graph signals},
  author={Tanaka, Yuichi},
  journal={IEEE Transactions on Signal Processing},
  volume={66},
  number={14},
  pages={3752--3767},
  year={2018},
  publisher={IEEE}
}

@inproceedings{Narang2011,
  title={Downsampling graphs using spectral theory},
  author={Narang, Sunil K and Ortega, Antonio},
  booktitle={2011 IEEE International Conference on Acoustics, Speech and Signal Processing (ICASSP)},
  pages={4208--4211},
  year={2011},
  organization={IEEE}
}

@incollection{Lorenzo2018,
  title={Sampling and recovery of graph signals},
  author={Lorenzo, PaoloDi and Barbarossa, Sergio and Banelli, Paolo},
  booktitle={Cooperative and Graph Signal Processing},
  pages={261--282},
  year={2018},
  publisher={Elsevier}
}

@article{Narang2012,
  title={Perfect reconstruction two-channel wavelet filter banks for graph structured data},
  author={Narang, Sunil K and Ortega, Antonio},
  journal={IEEE Transactions on Signal Processing},
  volume={60},
  number={6},
  pages={2786--2799},
  year={2012},
  publisher={IEEE}
}

@incollection{Tremblay2018,
  title={Design of graph filters and filterbanks},
  author={Tremblay, Nicolas and Gon{\c{c}}alves, Paulo and Borgnat, Pierre},
  booktitle={Cooperative and Graph Signal Processing},
  pages={299--324},
  year={2018},
  publisher={Elsevier}
}

@article{Sakiyama2014,
  title={Oversampled graph Laplacian matrix for graph filter banks},
  author={Sakiyama, Akie and Tanaka, Yuichi},
  journal={IEEE Transactions on Signal Processing},
  volume={62},
  number={24},
  pages={6425--6437},
  year={2014},
  publisher={IEEE}
}

@inproceedings{Anis2017,
  title={Critical sampling for wavelet filterbanks on arbitrary graphs},
  author={Anis, Aamir and Ortega, Antonio},
  booktitle={2017 IEEE International Conference on Acoustics, Speech and Signal Processing (ICASSP)},
  pages={3889--3893},
  year={2017},
  organization={IEEE}
}

@inproceedings{Narang2009,
  title={Lifting based wavelet transforms on graphs},
  author={Narang, Sunil K and Ortega, Antonio},
  booktitle={Proceedings: Asia-Pacific Signal and Information Processing Association, 2009 Annual Summit and Conference},
  pages={441--444},
  year={2009},
  organization={Asia-Pacific Signal and Information Processing Association}
}

@inproceedings{Crovella2003,
  title={Graph wavelets for spatial traffic analysis},
  author={Crovella, Mark and Kolaczyk, Eric},
  booktitle={IEEE INFOCOM 2003. Twenty-second Annual Joint Conference of the IEEE Computer and Communications Societies},
  volume={3},
  pages={1848--1857},
  year={2003},
  organization={IEEE}
}

@article{Hammond2011,
  title={Wavelets on graphs via spectral graph theory},
  author={Hammond, David K and Vandergheynst, Pierre and Gribonval, R{\'e}mi},
  journal={Applied and Computational Harmonic Analysis},
  volume={30},
  number={2},
  pages={129--150},
  year={2011},
  publisher={Elsevier}
}

@article{Bulai2023,
  title={Spectral graph wavelet packets frames},
  author={Bulai, Iulia Martina and Saliani, Sandra},
  journal={Applied and Computational Harmonic Analysis},
  volume={66},
  pages={18--45},
  year={2023},
  publisher={Elsevier}
}

@article{Stankovic2019,
  title={Graph signal processing--Part II: Processing and analyzing signals on graphs},
  author={Stankovic, Ljubisa and Mandic, Danilo and Dakovic, Milos and Brajovic, Milos and Scalzo, Bruno and Constantinides, Anthony G},
  journal={arXiv preprint arXiv:1909.10325},
  year={2019}
}

@article{Waheed2018,
  title={Graph polynomial filter for signal denoising},
  author={Waheed, Waseem and Tay, David BH},
  journal={IET Signal Processing},
  volume={12},
  number={3},
  pages={301--309},
  year={2018},
  publisher={Wiley Online Library}
}

@article{Onuki2016,
  title={Graph signal denoising via trilateral filter on graph spectral domain},
  author={Onuki, Masaki and Ono, Shunsuke and Yamagishi, Masao and Tanaka, Yuichi},
  journal={IEEE Transactions on Signal and Information Processing over Networks},
  volume={2},
  number={2},
  pages={137--148},
  year={2016},
  publisher={IEEE}
}

@book{Priestley1982,
  title={Spectral Analysis and Time Series, Two-Volume Set: Volumes I and II},
  author={Priestley, M.B.},
  isbn={9780125649223},
  lccn={80040235},
  year={1982},
  publisher={Elsevier Science}
}

@book{Von2002,
  title={Statistical Analysis in Climate Research},
  author={Von Storch, Hans and Zwiers, Francis W},
  year={2002},
  publisher={Cambridge University Press}
}

@book{Fuller2009,
  title={Introduction to Statistical Time Series},
  author={Fuller, Wayne A},
  year={2009},
  publisher={John Wiley \& Sons}
}

@article{Zachary1977,
  title={An information flow model for conflict and fission in small groups},
  author={Zachary, Wayne W},
  journal={Journal of Anthropological Research},
  volume={33},
  number={4},
  pages={452--473},
  year={1977},
  publisher={University of New Mexico}
}

@book{Phillips2006,
  title={Interpolation and Approximation by Polynomials},
  author={Phillips, George M},
  year={2006},
  publisher={Springer Science \& Business Media}
}

@article{Kim2023,
  title={Network time series forecasting using spectral graph wavelet transform},
  author={Kim, Kyusoon and Oh, Hee-Seok},
  journal={International Journal of Forecasting},
  year={2024},
  publisher={Elsevier}
}

@article{Canbolat2024,
  title={Locally stationary graph processes},
  author={Canbolat, Abdullah and Vural, Elif},
  journal={IEEE Transactions on Signal Processing},
   volume={72},
  pages={2323--2332},
  year={2024},
  publisher={IEEE}
}

@article{Loukas2019,
  title={Stationary time-vertex signal processing},
  author={Loukas, Andreas and Perraudin, Nathana{\"e}l},
  journal={EURASIP Journal on Advances in Signal Processing},
  volume={2019},
  number={1},
  pages={1--19},
  year={2019},
  publisher={Springer}
}

@article{Wei2024generalized,
	title={Generalized sampling of multi-dimensional graph signals based on prior information},
	author={Wei, Deyun and Yan, Zhenyang},
	journal={Signal Processing},
	volume={224},
	pages={109601},
	year={2024},
	publisher={Elsevier}
}

@article{Yang2022recommendation,
	title={A recommendation prediction method based on the estimation of PSD of sampled signals on graph},
	author={Yang, Zhihua and Kuang, Zhonghui and Yang, Lihua and Zhang, Qian},
	journal={Expert Systems with Applications},
	volume={201},
	pages={117097},
	year={2022},
	publisher={Elsevier}
}

@inproceedings{Giraldo2021emerging,
	title={The emerging field of graph signal processing for moving object segmentation},
	author={Giraldo, Jhony H and Javed, Sajid and Sultana, Maryam and Jung, Soon Ki and Bouwmans, Thierry},
	booktitle={International workshop on frontiers of computer vision},
	pages={31--45},
	year={2021},
	organization={Springer}
}

@article{Deng2024body,
	title={Body Motion Segmentation via Multilayer Graph Processing for Wearable Sensor Signals},
	author={Deng, Qinwen and Zhang, Songyang and Ding, Zhi},
	journal={IEEE Open Journal of Signal Processing},
	year={2024},
	publisher={IEEE}
}

@article{Giraldo2022graph,
	title={Graph moving object segmentation},
	author={Giraldo, Jhony H and Javed, Sajid and Bouwmans, Thierry},
	journal={IEEE Transactions on Pattern Analysis and Machine Intelligence},
	volume={44},
	number={5},
	pages={2485--2503},
	year={2022},
	publisher={IEEE}
}

@article{Liu2024point,
	title={Point cloud attacks in graph spectral domain: When 3d geometry meets graph signal processing},
	author={Liu, Daizong and Hu, Wei and Li, Xin},
	journal={IEEE Transactions on Pattern Analysis and Machine Intelligence},
	year={2024},
	publisher={IEEE}
}

@article{Wei2022geometry,
	author={Hu, Wei and Pang, Jiahao and Liu, Xianming and Tian, Dong and Lin, Chia-Wen and Vetro, Anthony},
	journal={IEEE Transactions on Multimedia}, 
	title={Graph Signal Processing for Geometric Data and Beyond: Theory and Applications}, 
	year={2022},
	volume={24},
	number={},
	pages={3961-3977}}

@article{Bahonar2021graph,
	title={Graph embedding using frequency filtering},
	author={Bahonar, Hoda and Mirzaei, Abdolreza and Sadri, Saeed and Wilson, Richard C},
	journal={IEEE transactions on pattern analysis and machine intelligence},
	volume={43},
	number={2},
	pages={473--484},
	year={2021},
	publisher={IEEE}
}

@article{Zhou2024graph,
	title={Graph Denoising with Framelet Regularizers},
	author={Zhou, Bingxin and Li, Ruikun and Zheng, Xuebin and Wang, Yu Guang and Gao, Junbin},
	journal={IEEE Transactions on Pattern Analysis and Machine Intelligence},
	year={2024},
	publisher={IEEE}
}

@article{Yang2022inferring,
	title={Inferring point cloud quality via graph similarity},
	author={Yang, Qi and Ma, Zhan and Xu, Yiling and Li, Zhu and Sun, Jun},
	journal={IEEE transactions on pattern analysis and machine intelligence},
	volume={44},
	number={6},
	pages={3015--3029},
	year={2022},
	publisher={IEEE}
}

@book{Horn1985, 
	place={Cambridge}, 
	title={Matrix Analysis}, 
	publisher={Cambridge University Press}, 
	author={Horn, Roger A. and Johnson, Charles R.}, 
	year={1985}
}

@article{Fiecas2014data,
	title={Data-driven shrinkage of the spectral density matrix of a high-dimensional time series},
	author={Fiecas, Mark and von Sachs, Rainer},
	year={2014}
}

@article{Botts2006shrinkage,
	title={A shrinkage estimator for spectral densities},
	author={Botts, Carsten H and Daniels, Michael J},
	journal={Biometrika},
	volume={93},
	number={1},
	pages={179--195},
	year={2006},
	publisher={Oxford University Press}
}

@article{Abonazel2022developing,
	title={Developing robust ridge estimators for Poisson regression model},
	author={Abonazel, Mohamed R and Dawoud, Issam},
	journal={Concurrency and Computation: Practice and Experience},
	volume={34},
	number={15},
	pages={e6979},
	year={2022},
	publisher={Wiley Online Library}
}

\end{document}